\documentclass[prepint,journal]{IEEEtran}

\usepackage{cite}
\usepackage{amsmath,amssymb,amsfonts,amsthm}
\usepackage{graphicx}
\usepackage{textcomp}
\usepackage{xcolor}
\usepackage{comment}
\usepackage{float}
\usepackage{tikz}
\usetikzlibrary{quantikz2}
\usetikzlibrary{math}
\usetikzlibrary{arrows,arrows.meta}
\usetikzlibrary{shapes.geometric}
\tikzcdset{arrow style=tikz, diagrams={>= latex}}
\usepackage[colorlinks=true]{hyperref}
\usepackage{enumerate}
\usepackage[linesnumbered,ruled,vlined]{algorithm2e}
\usepackage{array}
\usepackage[caption=false,font=normalsize,labelfont=sf,textfont=sf]{subfig}
\usepackage{stfloats}
\usepackage{url}
\usepackage{cite}
\newtheorem{theorem}{Theorem}

\newtheorem{proposition}{Proposition}
\newtheorem{lemma}{Lemma}
\newtheorem{fact}{Fact}

\newtheorem{definition}{Definition}
\newtheorem{example}{Example}
\newtheorem{remark}{Remark}

\begin{document}

\title{A complete set of transformation rules for reversible circuits}

\author{Shiguang Feng,~\IEEEmembership{Memeber,~IEEE,}
	Lvzhou Li
	\thanks{This work was supported by the National Key Research and Development Program of China (Grant No.2024YFB4504004), the National Natural Science Foundation of China (Grant No. 92465202, 62272492, 12447107),  the Guangdong Provincial Quantum Science Strategic Initiative (Grant No. GDZX2303007), the Guangzhou Science and Technology Program (Grant No. 2024A04J4892). (Corresponding author: Lvzhou Li.)}
	\thanks{Shiguang Feng and Lvzhou Li are with the School of Computer Science and Engineering, Sun Yat-sen University, Guangzhou 510006, China (e-mail:fengshg3@mail.sysu.edu.cn; lilvzh@mail.sysu.edu.cn).}}

\markboth{IEEE Transactions on Computer-Aided Design of Integrated Circuits and Systems,~Vol.~XX, No.~X, XXXX}{Shiguang Feng, Lvzhou Li: A complete set of transformation rules for reversible circuits}

\IEEEpubid{Citation information: DOI 10.1109/TCAD.2025.3641533~\copyright~2025 IEEE}

\maketitle

\begin{abstract}
Reversible logic synthesis is a crucial component in quantum electronic design automation. While rule-based methodologies have gained prominence in reversible circuit optimization, the completeness of the transformation rule systems is a longstanding problem in this domain. In this work, we propose the first complete set of transformation rules for reversible circuits, comprising five fundamental rules: any two equivalent reversible circuits can be transformed into each other using the rules. To prove the completeness, a canonical circuit representation for reversible functions is introduced, and we show that every reversible function is computed by a unique reversible circuit in the canonical form, and any reversible circuit can be transformed into its canonical form by applying the rules. 
\end{abstract}


\section{Introduction}
Quantum computing is an emerging field that leverages the principles of quantum mechanics to solve problems beyond the capabilities of classical computers. The efficient execution of quantum algorithms is a prerequisite for achieving quantum computational advantage. As a critical step in quantum computing, quantum compilation transforms the high-level descriptions of quantum algorithms into the low-level executable quantum circuits that comply with the constraints of specific quantum hardware, which has become an indispensable component of quantum electronic design automation (EDA). In applications of EDA, only sound and complete axiomatizations are of interest~\cite{Amaru2017newmajority}. 
In 2023, Cl\'{e}ment et~al. first introduced a complete equational theory for quantum circuits through their seminal work~\cite{Clement2023complete}, settling a 30-year-old open problem. Subsequent research efforts have focused on the structural minimality and extensions within this framework~\cite{Clement2024minimal,Clement2024qauntum}. The axiomatization of quantum circuits achieved a pivotal breakthrough that resolves the gaps in the systematic understanding of quantum circuit algebraization and provides categorical completeness guarantees for verification protocols.

Reversible circuits constitute a principal subclass of quantum circuits, which Toffoli first introduced as a computation model for the reversible computational process~\cite{Toffoli1980reversible}. 
They are Turing-complete and are polynomially equivalent to classical Boolean circuits within computational complexity theory.
Due to the inherent reversibility of quantum operations, any classical algorithm that needs to run on quantum computers must be converted into a reversible circuit. This transformation enables the exploitation of quantum phenomena, such as superposition and entanglement, to address complex classical problems through quantum computation. Some prominent quantum algorithms incorporate reversible circuits as core components, such as the oracle in Grover's search algorithm and the modular exponentiation module in Shor's factoring algorithm.

The realization of reversible functions is a very challenging task in quantum algorithm design. Reversible logic synthesis generates an optimized reversible circuit from a functional specification. As a fundamental component of reversible logic synthesis, circuit optimization improves the executability of quantum algorithms by reducing the circuit size and depth, and has garnered substantial research~\cite{Saeedi2013synthesis,Abdessaied2016reversiblequantum,Zulehner2021introducing,Wu2024asymptotically}. The rule-based and template-based methods are widely used for circuit optimization. A transformation (or rewriting) rule consists of a pair of equivalent circuits.
By applying the rules, we can transform a reversible circuit into a smaller equivalent circuit. Numerous rule-based optimization methods have been developed~\cite{Arabzadeh2010rule,Soeken2010window,Soeken2012optimizing,Cheng2012simplification,Rahman2014templates,Datta2015post,Bernardino2025reversible}. 
Templates are a generalization of rules. A template is a sequence of gates $A_1A_2\cdots A_m$ that performs the identity function. If a reversible circuit $\mathbf{C}$ contains a sequence $A_1A_2\cdots A_k$ ($k>m/2$) as its subcircuit, then $\mathbf{C}$ can be optimized by substituting $A_m A_{m-1}\cdots A_{k-1}$ for $A_1A_2\cdots A_k$ to reduce the number of gates. This process, known as template matching, has been intensively studied in the literature~\cite{Miller2003transformation,Maslov2003fredkin,Maslov2005toffoli,Maslov2005quantum,Vos2010reversible,Abdessaied2013exact,Taha2015reversible}.

The template-based and rule-based methods are essentially the same technique. A set of rules or templates is considered complete if any two equivalent circuits can be transformed into one another using those rules. Based on a complete set of templates, template matching can result in optimal circuits~\cite{Rahman2012properties}. However, the aforementioned optimization approaches are incomplete and cannot guarantee an optimal circuit after optimization~\cite{Abdessaied2016reversiblequantum}. 
In 2002, Iwama et~al. presented a complete set of transformation rules for reversible circuits that compute single-output Boolean functions~\cite{Iwama2002transrule}. Since then, the completeness of rule-based methods has attracted significant research interest in this field~\cite{Rahman2012properties,Soeken2013white,Thomsen2015ricerar,Hutslar2018library}. 
\IEEEpubidadjcol
Based on category theory, Cockett et~al. presented the complete sets of transformation rules for reversible circuits employing CNOT gates in 2017~\cite{Robin2017category}, and those employing Toffoli gates in 2018~\cite{Comfort2018category}, respectively.
The above works only discuss the complete set of transformation rules for some special reversible circuits. Whether there is a complete set of transformation rules for general reversible circuits is a longstanding problem. In particular, given that the existence of a complete set of transformation rules for quantum circuits has been proven~\cite{Clement2023complete}, it is more urgent to consider this problem for reversible circuits.

The main contribution of this work is the first complete set of transformation rules for reversible circuits. We systematically review and consolidate existing circuit transformation rules and optimization templates, and introduce a new rule to establish a set of five fundamental rules. To prove the completeness, we define a novel canonical circuit representation for $n$-ary reversible functions derived from a Hamiltonian path of an $n$-hypercube graph\footnote{A reversible function can be interpreted as a permutation over a Hamiltonian path of a hypercube graph.}. We show that every reversible function is computed by a unique reversible circuit in the canonical form, and any reversible circuit can be transformed into its canonical form. Hence, any two equivalent reversible circuits can be mutually transformed through the unique canonical form, via systematic application of the rules (see Fig.~\ref{fig-ideacomplete}). The transformation rules proposed in this paper provide formal guarantees for optimization completeness – any rule-based optimization approach subsuming the five rules can theoretically achieve circuit optimality. The developed theoretical framework establishes mathematical underpinnings for rule-based circuit optimization methodologies in quantum EDA systems.

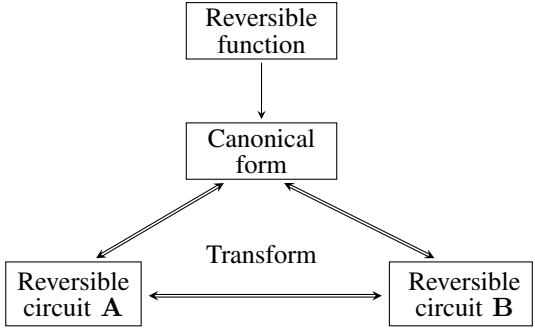
\begin{figure}[t]
	\centering
	\scalebox{0.8}{
	\begin{tikzpicture}
		\draw (1.6,-0.65) rectangle (3.4,-1.5);
		\node at (2.5,-0.9) [] {Reversible};
		\node at (2.5,-1.25) [] {circuit $\mathbf{A}$};
		
		\draw (4,0.45) rectangle (6,1.2) ;
		\node at (5,1) [] {Canonical};
		\node at (5,0.65) [] {form};
		
		\node at (5,-0.55) [] {Transform};
		
		\draw (6.7,-0.65) rectangle (8.6,-1.5);
		\node at (7.7,-0.9) [] {Reversible};
		\node at (7.7,-1.25) [] {circuit $\mathbf{B}$};
		
		\draw (4,2.05) rectangle (6,2.8);
		\node at (5,2.6) [] {Reversible};
		\node at (5,2.25) [] {function};
		
		\draw[>= stealth, ->, black] (5,2) -- (5,1.25);
		
		\draw[>= stealth, double, <->, black] (3.5,-1.1) -- (6.6,-1.1);
		
		\draw[>= stealth, double, <->, black] (2.8,-0.6) -- (4.5,0.4);
		
		\draw[>= stealth, double, <->, black] (7.3,-0.6) -- (5.3,0.4);
	\end{tikzpicture}}
	\caption{Sketch of the completeness proof.}
	\label{fig-ideacomplete}
\end{figure}

The paper is organized as follows. In Section~\ref{sec-pre}, we set up the notation and terminology of reversible functions and reversible circuits. 
In Section~\ref{sec-transformationrules}, we propose a set $\mathcal{RC}$ of transformation rules for reversible circuits, and prove the soundness of $\mathcal{RC}$.
In Section~\ref{sec-completeness}, we define the canonical forms of reversible circuits and prove the completeness of $\mathcal{RC}$.
Finally, we conclude the paper in Section~\ref{sec-conclusion}.

\section{Preliminaries}\label{sec-pre}
An $n$-ary reversible function 
\[f(x_1,x_2\dots,x_n)=(y_1,y_2\dots,y_n)\]
where $x_i,y_i\in\{0,1\}$ ($1\leq i \leq n$), is a bijection from $\{0,1\}^n$ to $\{0,1\}^n$. 
A reversible logic gate computes a reversible function. The X, CNOT, and Toffoli gates are three elementary reversible logic gates (see Fig.~\ref{fig-XCNOTTOFFgate}). The X gate is a 1-bit reversible logic gate that flips the input bit. The CNOT gate is a 2-bit reversible logic gate that has one control bit and one target bit. It flips the target bit iff the control bit has value 1. The Toffoli gate is a 3-bit reversible logic gate that has two control bits and one target bit: it flips the target bit iff both of the two control bits have value 1.

\begin{figure}[t]
	\centering
	\subfloat{
	\scalebox{0.8}{
	\begin{quantikz}
			\lstick{$x$} & \gate{X} & \rstick{$\neg x$}
	\end{quantikz}}}
	\subfloat{
	\scalebox{0.8}{
	\begin{quantikz}
		\lstick{$x$} & \ctrl{1} & \rstick{$x$} \\
		\lstick{$y$} & \targ{} & \rstick{$x \oplus y$}
	\end{quantikz}}}
	\subfloat{
	\scalebox{0.8}{
	\begin{quantikz}
		\lstick{$x$} & \ctrl{1} & \rstick{$x$}  \\
		\lstick{$y$} & \ctrl{1} & \rstick{$y$} \\
		\lstick{$z$} & \targ{} & \rstick{$(x \land y) \oplus z$}
	\end{quantikz}}}
	\caption{\label{fig-XCNOTTOFFgate} The X gate, CNOT gate, and Toffoli gate.}
\end{figure}
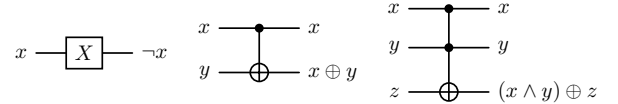


\begin{figure}[t]
	\centering
	\subfloat[\label{fig-gate-mct}]{
	\scalebox{0.8}{
	\begin{quantikz}[row sep={0.7cm,between origins}]
		\lstick{$x_1$}  & \ctrl{1}   & \rstick{$x_1$}  \\
		\lstick{$x_2$}  & \ctrl{2}   & \rstick{$x_2$}  \\
		\lstick{$\vdots$} & \wave  	 & \rstick{$\vdots$} \\
		\lstick{$x_n$} 	& \ctrl{1} 	 & \rstick{$x_n$} \\
		\lstick{$y$} & \targ{} & \rstick{$(\bigwedge_{i=1}^{n}x_{i})\oplus y$}
	\end{quantikz}}}
	\subfloat[\label{fig-gate-mpmct}]{
	\scalebox{0.8}{
	\begin{quantikz}[row sep={0.7cm,between origins}]
	\lstick{$x_1$}    & \octrl{1}    & \rstick{$x_1$}   \\
	\lstick{$x_2$}    & \ctrl{2}     & \rstick{$x_2$}   \\
	\lstick{$\vdots$} & \wave  	     & \rstick{$\vdots$} \\
	\lstick{$x_n$} 	  & \octrl{1}     & \rstick{$x_n$}  \\
	\lstick{$y$}      & \targ{}      & \rstick{$(\neg x_1\land x_2\land \cdots\land \neg x_n)\oplus y$} 
	\end{quantikz}}}
	\caption{\label{fig-gate}The illustration of (a) an MCT gate, and (b) an MPMCT gate.}
\end{figure}
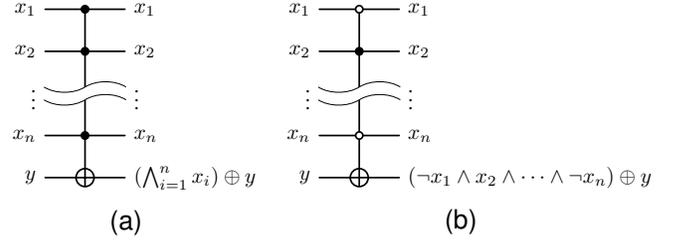

The mixed polarity multiple-control Toffoli (MPMCT) gates extend MCT gates with negative control bits. More precisely, an MPMCT gate has a set $P$ of positive control bits, a set $N$ of negative control bits, and a target bit (see Fig.~\ref{fig-gate-mpmct}, where the black dots and white dots denote the positive control bits and negative control bits, respectively). It flips the target bit iff all bits in $P$ have value 1 and all bits in $N$ have value 0.

Let $P,N$ be two sets of bits satisfying $P\cap N = \emptyset$, and $q$ a bit such that $q\notin P\cup N$. We use $\mathbf{G}[P,N,q]$ to denote the reversible logic gate $A$ where $P$ (resp. $N$) is the set of positive (resp. negative) control bits of $A$, and $q$ is the target bit of $A$. Hence, $\mathbf{G}[\emptyset,\emptyset,q]$ denotes the X gate that operates on $q$, $\mathbf{G}[\{p\},\emptyset,q]$ denotes the CNOT gate whose control (resp. target) bit is $p$ (resp. $q$), and $\mathbf{G}[P,\emptyset,q]$ denotes an MCT gate if $P$ has more than one element.
We abbreviate $\mathbf{G}[\emptyset,\emptyset,q]$ and $\mathbf{G}[\{p\},\emptyset,q]$ to $\mathrm{X}[q]$ and $\mathrm{CNOT}[p,q]$, respectively.

A reversible circuit is a sequence of reversible logic gates. We use the convention that the leftmost gate in the reversible circuit executes first. Toffoli showed that the X, CNOT, Toffoli, and MCT gates are universal~\cite{Toffoli1980reversible}. For technical convenience, we also allow MPMCT gates. An ancillary bit is an extra input bit being used as the temporary workspace. It is usually initialized to 0. Since every $n$-ary reversible function can be computed by an $n$-bit reversible circuit~\cite{Toffoli1980reversible}, we focus on the reversible circuits without ancillary bits in this paper. Unless otherwise stated, for an $n$-bit reversible logic gate or $n$-bit reversible circuit, we assume that it operates on the bits $\{q_0,\dots,q_{n-1}\}$. 

\begin{example}
	The following is a picture visualization of the circuit
	\[
	\begin{aligned}
	\mathbf{C}	= \bigl(\mathrm{CNOT}[q_0,q_2] \mathrm{X}[q_3] \mathbf{G}[\{q_0,q_2\},\{q_1\},q_3] \\
		\mathrm{CNOT}[q_2,q_1] \mathbf{G}[\emptyset,\{q_0,q_1\},q_2]\mathbf{G}[\{q_2,q_3\},\emptyset,q_1]\bigr)
	\end{aligned}
	\]
	\[
	\mathbf{C}	=
	\scalebox{0.7}{
	\begin{quantikz}[row sep={0.7cm,between origins}]
			\lstick{$q_0:$} & \ctrl{2} & \ctrl{1}  & \qw  		& \octrl{1} & \qw    	& \qw \\
			\lstick{$q_1:$} & \qw	   & \octrl{1} & \targ{} 	& \octrl{1}	& \targ{}   & \qw \\
			\lstick{$q_2:$} & \targ{1} & \ctrl{1}  & \ctrl{-1} 	& \targ{}   & \ctrl{-1}	& \qw \\
			\lstick{$q_3:$} & \gate{X} & \targ{1}  & \qw 		& \qw 		& \ctrl{-1}	& \qw 
	\end{quantikz}}
	\]
\end{example}

We say that an $n$-ary reversible function $f$ exchanges two strings $a,b$ if
\[
 f(x) = 
\begin{cases} 
	a, & \text{if } x = b, \\
	b, & \text{if } x = a, \\
	x, & \text{if } x\notin \{a,b\}.
\end{cases}
\]
A reversible circuit $\mathbf{C}$ exchanges $a,b$ if the reversible function computed by $\mathbf{C}$ exchanges $a,b$.
Let $\mathbf{A}$ and $\mathbf{B}$ be two reversible circuits.  We use $\mathbf{A}\equiv \mathbf{B}$ to denote that $\mathbf{A}$ and $\mathbf{B}$ are equivalent, i.e., they compute the same reversible function, and use $\mathbf{AB}$ to denote the reversible circuit that is a concatenation of $\mathbf{A}$ and $\mathbf{B}$.

\section{Transformation rules}\label{sec-transformationrules}
In this section, we propose a set of transformation rules for reversible circuits and prove their soundness.
 
\subsection{The set of transformation rules}
We define $\mathcal{RC}$ to be the set of the following five basic transformation rules.

\begin{enumerate}[\textbf{Rule} 1.]
	\item\label{rule-gate-elimn} For any reversible logic gate $A$, 
	\[AA \equiv \epsilon\]
	where $\epsilon$ denotes the empty circuit. 

	\item\label{rule-qubit-add-remove} If $A_0 = \mathbf{G}[P,N\cup \{p\},q]$, $A_1 = \mathbf{G}[P\cup \{p\},N,q]$, and $A =\mathbf{G}[P,N,q]$, then
	\[
	A_0 A_1 \equiv A.
	\]
	
	\item\label{rule-gateswap} If $A =\mathbf{G}[P_1,N_1,p]$, $B =\mathbf{G}[P_2,N_2,q]$ are two gates satisfying $P_1\cap N_2\neq \emptyset$ or $P_2\cap N_1\neq \emptyset$, then
	\[
	A B \equiv B A.
	\]
	
	\item\label{rule-swap-eliman} If $A=\mathrm{CNOT}[p,q]$, $B=\mathrm{CNOT}[q,p]$, and
	\[
	\begin{aligned}
		C_1 & = \mathbf{G}[P\cup P_1,N\cup N_1,p], \\
		C_2 & = \mathbf{G}[P\cup P_2, N\cup N_2, q]
	\end{aligned}
	\]
	are four gates in which the sets $P_1,P_2,N_1,N_2$ satisfy one of the following conditions:
	\begin{itemize}
		\item $P_1 =\{q\}$, $P_2 =\{p\}$, $N_1 = N_2=\emptyset$,
		\item $N_1 =\{q\}$, $N_2 =\{p\}$, $P_1 = P_2=\emptyset$,
	\end{itemize}
	then
	\[
	ABA C_1 \equiv C_2 ABA.
	\]
	
	\item\label{rule-completeness} Let $A_0 = \mathbf{G}[P,N\cup Q, q_0]$, $A_1 = \mathbf{G}[P\cup Q,N,q_0]$, where $Q=\{q_1,\dots,q_m\}$. 
	Set $P'=P\cup\{q_0\}$ and $N'=N\cup\{q_0\}$. For each $1\leq i \leq m$, define 
	\[
	\begin{aligned}
		B_i = &\mathbf{G}[P'\cup\{q_{i+1},\dots,q_m\},N\cup \{q_1,\dots,q_{i-1}\}, q_i] \\
		B_i' =&\mathbf{G}[P\cup\{q_{i+1},\dots,q_m\},N'\cup \{q_1,\dots,q_{i-1}\}, q_i].
	\end{aligned}
	\]
	Then
	\[
	A_0 A_1 B_1\cdots B_m\cdots B_1 A_1 A_0 \equiv B_1'\cdots B_m'\cdots B_1'.
	\]
\end{enumerate}

Let's briefly explain the five rules with examples. Rule~\ref{rule-gate-elimn} says that two adjacent identical gates can be removed from the circuit. 
Rule~\ref{rule-qubit-add-remove} says that if two adjacent gates have the same control bits where exactly one bit $p$ among them has different polarities in the two gates, then the two gates can be reduced to one gate with $p$ removed, as shown in the following example.
\[
	\scalebox{0.7}{
	\begin{quantikz}[row sep={0.7cm,between origins}]
		& \octrl{1}  & \ctrl{1}&  \qw \\
		& \ctrl{2}   & \ctrl{2} &  \qw \\
		& \wave 	 &          &  \qw \\
		& \ctrl{1}   & \ctrl{1} &  \qw \\
		& \targ{}    & \targ{}  &  \qw
	\end{quantikz}}
	\equiv
	\scalebox{0.7}{
	\begin{quantikz}[row sep={0.7cm,between origins}]
		& \qw 		&  \qw \\
		& \ctrl{2}  &  \qw \\
		& \wave     &  \qw \\
		& \ctrl{1}  &  \qw \\
		& \targ{}   &  \qw
	\end{quantikz}}
\]

Rule~\ref{rule-gateswap} says that two gates commute if they have a common control bit that has different polarities in the two gates, respectively. For example,
\[
	\scalebox{0.7}{
	\begin{quantikz}[row sep={0.7cm,between origins}]
		& \ctrl{1}   & \qw  	 &  \qw \\
		& \ctrl{2}   & \octrl{2} &  \qw \\
		& \wave 	 &           &  \qw \\
		& \ctrl{1}   & \targ{}   &  \qw \\
		& \targ{}    & \ctrl{-1} &  \qw
	\end{quantikz}}
	\equiv
	\scalebox{0.7}{
	\begin{quantikz}[row sep={0.7cm,between origins}]
		& \qw		 & \ctrl{1}   &  \qw \\
		& \octrl{2}  & \ctrl{2}   &  \qw \\
		& \wave 	 &            &  \qw \\
		& \targ{}    & \ctrl{1}   &  \qw \\
		& \ctrl{-1}  & \targ{}    &  \qw
	\end{quantikz}}
\]

The SWAP gate is a widely used gate in quantum circuits that swaps the states of two quantum bits. It can be decomposed into three CNOT gates, as shown below.
\[	
	\scalebox{0.7}{
	\begin{quantikz}[row sep={0.6cm,between origins}]
		\lstick{$x$} & \ctrl{1} & \targ{}   & \ctrl{1} & \rstick{$y$} \\
		\lstick{$y$} & \targ{}  & \ctrl{-1} & \targ{}  & \rstick{$x$}
	\end{quantikz}}
	=
	\scalebox{0.7}{
	\begin{quantikz}[row sep={0.6cm,between origins}]
		\lstick{$x$} & \swap{1}  & \rstick{$y$} \\
		\lstick{$y$} & \targX{}  & \rstick{$x$}
	\end{quantikz}}
\]
We do not regard the SWAP gate as an elementary gate, but rather consider its equivalent form in the paper. Roughly speaking, by Rule~\ref{rule-swap-eliman}, if a SWAP gate operates on the bits of a gate $A$, then we can move the SWAP gate through $A$ with its corresponding control and target bits exchanged. The following is an example of Rule~\ref{rule-swap-eliman}.
	\[
	\scalebox{0.7}{
	\begin{quantikz}[row sep={0.7cm,between origins}]
		& \qw  	   & \qw 	   & \qw		& \ctrl{2}  & \qw \\
		& \qw      & \qw 	   & \qw		& \wave	   	& \qw \\
		& \qw  	   & \qw 	   & \qw		& \octrl{1} & \qw \\
		& \ctrl{1} & \targ{}   & \ctrl{1}   & \targ{}   & \qw \\
		& \targ{}  & \ctrl{-1} & \targ{}    & \ctrl{-1} & \qw 
	\end{quantikz}}
	\equiv
	\scalebox{0.7}{
	\begin{quantikz}[row sep={0.7cm,between origins}]
		& \ctrl{2}  & \qw  	   & \qw 	   & \qw	  & \qw \\
		& \wave		& \qw      & \qw 	   & \qw	  & \qw \\
		& \octrl{1}	& \qw  	   & \qw 	   & \qw	  & \qw \\
		& \ctrl{1}  & \ctrl{1} & \targ{}   & \ctrl{1} & \qw \\
		& \targ{}	& \targ{}  & \ctrl{-1} & \targ{}  & \qw 
	\end{quantikz}}
	\]

Rule~\ref{rule-completeness} can be used to change the polarity of control bits. It is essential for the proof of the completeness of $\mathcal{RC}$. The following is a specific instance of this rule, where $Q=\{q_1,q_2,q_3\}$, $P=N=\emptyset$, $A_0 = \mathbf{G}[\emptyset, Q, q_0]$, $A_1 = \mathbf{G}[Q, \emptyset, q_0]$, and
\[
\begin{aligned}
	B_1 =  \mathbf{G}[\{q_0,q_2,q_3\}, \emptyset, q_1],  	B_1' = \mathbf{G}[\{q_2,q_3\}, \{q_0\}, q_1], \\
	B_2 =  \mathbf{G}[\{q_0,q_3\}, \{q_1\}, q_2],  	B_2' = \mathbf{G}[\{q_3\}, \{q_0,q_1\}, q_2], \\
	B_3 =  \mathbf{G}[\{q_0\}, \{q_1,q_2\}, q_3],  	B_3' = \mathbf{G}[\emptyset, \{q_0,q_1,q_2\}, q_3].
\end{aligned}
\]

\[
	\begin{aligned}
		 & \scalebox{0.7}{\begin{quantikz}[row sep={0.7cm,between origins}]
		 	\setwiretype{n}	& \push{A_0}& \push{A_1}& \push{B_1} & \push{B_2}	& \push{B_3} & \push{B_2}	& \push{B_1}& \push{A_1} &\push{A_0}  & \\
			\lstick{$q_3$:} & \octrl{1} & \ctrl{1} & \ctrl{1}  & \ctrl{1}	& \targ{}	 & \ctrl{1} 	& \ctrl{1} 	& \ctrl{1} & \octrl{1}  & \qw \\
			\lstick{$q_2$:} & \octrl{1} & \ctrl{1} & \ctrl{1}  & \targ{}	& \octrl{-1} & \targ{}		& \ctrl{1} 	& \ctrl{1} & \octrl{1}  & \qw \\
			\lstick{$q_1$:} & \octrl{1} & \ctrl{1} & \targ{}   & \octrl{-1}	& \octrl{-1} & \octrl{-1}	& \targ{}  	& \ctrl{1} & \octrl{1}  & \qw \\
			\lstick{$q_0$:} 	& \targ{}   & \targ{}  & \ctrl{-1} & \ctrl{-1}	& \ctrl{-1}  & \ctrl{-1}	& \ctrl{-1} & \targ{}  & \targ{}    & \qw
		\end{quantikz}} \\
		& \qquad\equiv
		\scalebox{0.7}{\begin{quantikz}[row sep={0.7cm,between origins}]
			\setwiretype{n}	& \push{B_1'} & \push{B_2'}	& \push{B_3'} & \push{B_2'}	& \push{B_1'}   &  \\
			\lstick{$q_3$:} & \ctrl{1}  & \ctrl{1}	& \targ{}	 & \ctrl{1} 	& \ctrl{1} 	  & \qw \\
			\lstick{$q_2$:} & \ctrl{1}  & \targ{}	& \octrl{-1} & \targ{}		& \ctrl{1}   & \qw \\
			\lstick{$q_1$:} & \targ{}   & \octrl{-1}	& \octrl{-1} & \octrl{-1}	& \targ{}  	 & \qw \\
			\lstick{$q_0$:} 	& \octrl{-1} & \octrl{-1}	& \octrl{-1}  & \octrl{-1}	& \octrl{-1}  & \qw
		\end{quantikz}}
	\end{aligned}
\]

\begin{example}
	An application of Rule~\ref{rule-completeness} is transforming a CNOT gate to a negatively controlled-NOT gate, as shown below. We first expand the two X gates to two subcircuits $A_0 A_1$ by Rule~\ref{rule-qubit-add-remove}, then apply Rule~\ref{rule-gateswap} to transform the second subcircuit $A_0 A_1$ to $A_1 A_0$ by swapping $A_0, A_1$, and finally apply Rule~\ref{rule-completeness} to the whole circuit to obtain the negatively controlled-NOT gate.
	\[
	\begin{aligned}
		&\scalebox{0.7}{
			\begin{quantikz}[row sep={0.7cm,between origins}]
				\setwiretype{n}	& & \push{B_1}  & & \\
				\lstick{$q_1$:} & \qw		& \targ{}    & \qw 		 & \qw \\
				\lstick{$q_0$:} & \gate{X}	& \ctrl{-1} & \gate{X}	 & \qw 
		\end{quantikz}} \overset{\text{Rule~\ref{rule-qubit-add-remove}}}{\equiv}
		\scalebox{0.7}{
			\begin{quantikz}[row sep={0.7cm,between origins}]
				\setwiretype{n}	& \push{A_0}& \push{A_1}& \push{B_1} & \push{A_0} &\push{A_1}  & \\
				\lstick{$q_1$:} & \octrl{1}& \ctrl{1} & \targ{}   & \octrl{1}  & \ctrl{1}  & \qw \\
				\lstick{$q_0$:} & \targ{}  & \targ{}   & \ctrl{-1}& \targ{}   & \targ{} & \qw 
		\end{quantikz}} \\
		& \overset{\text{Rule~\ref{rule-gateswap}}}{\equiv} \scalebox{0.7}{
			\begin{quantikz}[row sep={0.7cm,between origins}]
				\setwiretype{n}	& \push{A_0}& \push{A_1}& \push{B_1} & \push{A_1} &\push{A_0}  & \\
				\lstick{$q_1$:} & \octrl{1}& \ctrl{1} & \targ{}   & \ctrl{1}  & \octrl{1}  & \qw \\
				\lstick{$q_0$:} & \targ{}  & \targ{}   & \ctrl{-1}& \targ{}   & \targ{} & \qw 
		\end{quantikz}}
		\overset{\text{Rule~\ref{rule-completeness}}}{\equiv}
		\scalebox{0.7}{
			\begin{quantikz}[row sep={0.7cm,between origins}]
				\setwiretype{n}	& \push{B_1'}  & \\
				\lstick{$q_1$:} & \targ{}    & \qw \\
				\lstick{$q_0$:} & \octrl{-1}  & \qw 
		\end{quantikz}}
	\end{aligned}
	\]
	Actually, combining with other rules, we can transform every MPMCT gate into a combination of an MCT gate and X gates (see Rule~\ref{rule-MPMCTtoMCT}).
\end{example}

Let $\mathbf{A}$ and $\mathbf{B}$ be two reversible circuits. We use $\mathbf{A} \Leftrightarrow \mathbf{B}$ to denote that there is a transformation between $\mathbf{A}$ and $\mathbf{B}$ by applying the rules in $\mathcal{RC}$. Obviously, ``$\Leftrightarrow$'' is an equivalence relation. The following is easy to check by the definition.
\begin{itemize}
	\item\label{enum-prop-soundproperty1} If $\mathbf{A} \Leftrightarrow \mathbf{B}$ and $\mathbf{B} \Leftrightarrow \mathbf{C}$, then $\mathbf{A} \Leftrightarrow \mathbf{C}$.
	\item\label{enum-prop-soundproperty2} If $\mathbf{A} \Leftrightarrow \mathbf{B}$ and $\mathbf{C} \Leftrightarrow \mathbf{D}$, then $\mathbf{AC} \Leftrightarrow \mathbf{BD}$.
	\item\label{enum-prop-soundproperty3} If $\mathbf{A} \Leftrightarrow \mathbf{B}$, then $\mathbf{CAD} \Leftrightarrow \mathbf{CBD}$. 	
\end{itemize}

\begin{theorem}[\textbf{Soundness}]\label{thm-soundness}
	If $\mathbf{A} \Leftrightarrow \mathbf{B}$, then $\mathbf{A} \equiv \mathbf{B}$.
\end{theorem}
\begin{proof}
	It suffices to prove the soundness of the rules in $\mathcal{RC}$. It is easy to check for Rule~\ref{rule-gate-elimn} and \ref{rule-qubit-add-remove}. For the soundness of Rule~\ref{rule-gateswap}, note that if two gates $A$ and $B$ have a common control bit that has different polarities in them, then at most one gate works for an arbitrary input. Hence, changing the order of $A$ and $B$ does not influence the result of the computation.
	
	For the soundness of Rule~\ref{rule-swap-eliman}, let $Q=\{q_1,\dots,q_n\}$. Suppose that $A=\mathrm{CNOT}[q_i,q_j]$, $B=\mathrm{CNOT}[q_j,q_i]$, $C_1 = \mathbf{G}[Q/\{q_i\},\emptyset,q_i]$, and $C_2 = \mathbf{G}[Q/\{q_j\},\emptyset,q_j]$, where $1\leq i < j \leq n$. We show that the two circuits $ABA C_1$ and $C_2 ABA$ compute the same reversible function. The proof for the other case of Rule~\ref{rule-swap-eliman} is similar.
	
	It is easily seen that the circuit $ABA$ swaps the values of $q_i$ and $q_j$.	
	We use $\bar{s}=(s_1 \cdots s_n)\in \{0,1\}^n$ to denote that $s_k$ is the input of $q_k$ ($1\leq k \leq n$), and denote by $[\bar{s}]^{i}_{j}$ the sequence that swaps the $s_i$ and $s_j$ in $\bar{s}$.
	There are two cases need to consider.
	\begin{itemize}
		\item There exists some $q_k\in Q/\{q_j\}$ such that its input $s_k=0$. Hence, $C_2$ has no effect on $\bar{s}$. The result after executing the circuit $C_2 ABA$ on $\bar{s}$ is $[\bar{s}]^{i}_{j}$. Furthermore, $C_1$ has no effect on $[\bar{s}]^{i}_{j}$ either. The result after executing the circuit $ABA C_1$ on $\bar{s}$ is also $[\bar{s}]^{i}_{j}$.
		\item The inputs of all bits in $Q/\{q_j\}$ are 1. Executing the gate $C_2$ on $\bar{s}$ flips $s_j$, we denote the result by $\bar{s}'$. Then executing the circuit $ABA$ on $\bar{s}'$ we obtain $[\bar{s}']^{i}_{j}$. The result after executing the circuit $ABA$ on $\bar{s}$ is $[\bar{s}]^{i}_{j}$. Applying the gate $C_1$ on $[\bar{s}]^{i}_{j}$ we also obtain $[\bar{s}']^{i}_{j}$.
	\end{itemize}
	
	We now prove the soundness of Rule~\ref{rule-completeness}. Let $A_0$, $A_1$, $B_i$, and $B_i'$ ($1\leq i \leq m$) be the gates defined in the rule. There are two cases. 
	
	\textbf{Case\hspace{0.1cm}1:} $P=N=\emptyset$. We use the sequence $(s_0 s_1 \cdots s_m)\in \{0,1\}^{m+1}$ to denote that $s_k$ is the input of $q_k$ ($0\leq k \leq m$).
	Let $\bar{s}_1=(1 a_1\cdots a_m)$ $\bar{s}_2=(1 b_1\cdots b_m)$ where $a_i=0$ and $b_i=1$ ($1\leq i \leq m$). It is easy to verify that the circuit $B_1\cdots B_m\cdots B_1$ exchanges $\bar{s}_1=(10\cdots0)$ and $\bar{s}_2=(11\cdots1)$. That is, it outputs $\bar{s}_2$ (resp. $\bar{s}_1$) when given $\bar{s}_1$ (resp. $\bar{s}_2$), and outputs the same string when given $\bar{s}$ where $\bar{s}\neq \bar{s}_1$ and $\bar{s}\neq \bar{s}_2$. The gates $B_i$ and $B_i'$ ($1\leq i \leq m$) are only different on the polarity of control bit $q_0$, so the circuit $B_1'\cdots B_m'\cdots B_1'$ exchanges $(00\cdots0)$ and $(01\cdots1)$. 
	
	It is easily seen that the two circuits $A_0 A_1$ and $A_1 A_0$ are equivalent and they compute the reversible function
	\[
	f(x) = 
	\begin{cases} 
		00\cdots 0, & \text{if } x = 10\cdots 0, \\
		10\cdots 0, & \text{if } x = 00\cdots 0, \\
		01\cdots 1, & \text{if } x = 11\cdots 1, \\
		11\cdots 1, & \text{if } x = 01\cdots 1, \\
		x, & \text{if } x\notin \{00\cdots 0,10\cdots 0,01\cdots 1,11\cdots 1\}.
	\end{cases}
	\]
	Given a string $(00\cdots 0)$, $A_0 A_1$ converts it to $\bar{s}_1=(10\cdots 0)$, $B_1\cdots B_m\cdots B_1$ converts $\bar{s}_1$ to $\bar{s}_2=(11\cdots 1)$, and $A_1 A_0$ converts $\bar{s}_2$ to $(01\cdots 1)$. Similarly, for the other input string $\bar{s}$, by executing the circuits $A_0 A_1$, $B_1\cdots B_m\cdots B_1$, and $A_1 A_0$, we obtain $(00\cdots 0)$ if $\bar{s}=(01\cdots 1)$, and still $\bar{s}$ otherwise. Therefore, the circuit $A_0 A_1 B_1\cdots B_m\cdots B_1 A_1 A_0$ exchanges $(00\cdots0)$ and $(01\cdots1)$, and the following equivalence is proved
	\[
	A_0 A_1 B_1\cdots B_m\cdots B_1 A_1 A_0 \equiv B_1'\cdots B_m'\cdots B_1'.
	\]
	
	We give an example below to explain the proof. The following circuit $B_1' B_2' B_3' B_2' B_1'$ exchanges $(0000)$ and $(0111)$.
	\[	
	\scalebox{0.7}{\begin{quantikz}[row sep={0.7cm,between origins}]
			\setwiretype{n}	& \push{B_1'} & \push{B_2'}	& \push{B_3'} & \push{B_2'}	& \push{B_1'}   &  \\
			\lstick{$q_3$:}	& \ctrl{1} 	  & \ctrl{1}	& \targ{}	  & \ctrl{1}	& \ctrl{1} 		& \qw \\
			\lstick{$q_2$:}	& \ctrl{1}    & \targ{}		& \octrl{-1}  & \targ{}		& \ctrl{1}      & \qw \\
			\lstick{$q_1$:}	& \targ{}     & \octrl{-1}	& \octrl{-1}  & \octrl{-1}	& \targ{}  	 	& \qw \\
			\lstick{$q_0$:}	& \octrl{-1}  & \octrl{-1}	& \octrl{-1}  & \octrl{-1}	& \octrl{-1}  	& \qw
	\end{quantikz}}
	\]
	The circuit $ B_1 B_2 B_3 B_2 B_1$ exchanges $(1000)$ and $(1111)$.
	\[	
	\scalebox{0.7}{\begin{quantikz}[row sep={0.7cm,between origins}]
			\setwiretype{n}	& \push{B_1} & \push{B_2}	& \push{B_3} & \push{B_2}	& \push{B_1}   &  \\
			\lstick{$q_3$:}	& \ctrl{1} 	  & \ctrl{1}	& \targ{}	  & \ctrl{1}	& \ctrl{1} 		& \qw \\
			\lstick{$q_2$:}	& \ctrl{1}    & \targ{}		& \octrl{-1}  & \targ{}		& \ctrl{1}      & \qw \\
			\lstick{$q_1$:}	& \targ{}     & \octrl{-1}	& \octrl{-1}  & \octrl{-1}	& \targ{}  	 	& \qw \\
			\lstick{$q_0$:}	& \ctrl{-1}  & \ctrl{-1}	& \ctrl{-1}  & \ctrl{-1}	& \ctrl{-1}  	& \qw
	\end{quantikz}}
	\]
	The two circuits $A_0 A_1$ and $A_1 A_0$ compute the function
	\[
	g(x) = 
	\begin{cases} 
		0000, & \text{if } x = 1000, \\
		1000, & \text{if } x = 0000, \\
		0111, & \text{if } x = 1111, \\
		1111, & \text{if } x = 0111, \\
		x, & \text{if } x\notin \{0000,1000,0111,1111\}.
	\end{cases}
	\]
	By concatenating $A_0 A_1$, $B_1 B_2 B_3 B_2 B_1$, and $A_1 A_0$, we obtain the following circuit $A_0 A_1 B_1 B_2 B_3 B_2 B_1 A_1 A_0$ that exchanges $(0000)$ and $(0111)$, which is equivalent to  $B_1' B_2' B_3' B_2' B_1'$.
	\[
	\scalebox{0.75}{\begin{quantikz}[row sep={0.7cm,between origins}]
		\setwiretype{n}	& \push{A_0}& \push{A_1}& \push{B_1} & \push{B_2}	& \push{B_3} & \push{B_2}	& \push{B_1}& \push{A_1} &\push{A_0}  & \\
		\lstick{$q_3$:}	& \octrl{1} & \ctrl{1}  & \ctrl{1}   & \ctrl{1}		& \targ{}	 & \ctrl{1} 	& \ctrl{1} 	& \ctrl{1}   & \octrl{1}  & \qw \\
		\lstick{$q_2$:}	& \octrl{1} & \ctrl{1}  & \ctrl{1}   & \targ{}		& \octrl{-1} & \targ{}		& \ctrl{1} 	& \ctrl{1}   & \octrl{1}  & \qw \\
		\lstick{$q_1$:}	& \octrl{1} & \ctrl{1}  & \targ{}    & \octrl{-1}	& \octrl{-1} & \octrl{-1}	& \targ{}  	& \ctrl{1}   & \octrl{1}  & \qw \\
		\lstick{$q_0$:}	& \targ{}   & \targ{}   & \ctrl{-1}  & \ctrl{-1}	& \ctrl{-1}  & \ctrl{-1}	& \ctrl{-1} & \targ{}    & \targ{}    & \qw
	\end{quantikz}}
	\]
	
	\textbf{Case\hspace{0.1cm}2:} $P\neq \emptyset$ or $Q\neq \emptyset$. The input strings of the circuit can be divided into two sets:
	\begin{enumerate}[(i)]
		\item\label{enu-thm-sound1} the strings that assign 0 to a bit in $P$, or assign 1 to a bit in $N$;
		\item\label{enu-thm-sound2} the string that assign 1 to all bits in $P$, and assign 0 to all bits in $N$.
	\end{enumerate}
If the input is a string in \eqref{enu-thm-sound1}, then at least one control bit from $P\cup Q$ in the gates $A_0$, $A_1$, $B_i$, and $B_i'$ ($1\leq i \leq m$) has the value opposing its polarity. Hence, no gate flips its target bit, and the circuits on both sides of Rule~\ref{rule-completeness} are equivalent to an empty circuit. If the input is a string in \eqref{enu-thm-sound2}, then the control bits from $P\cup Q$ become irrelevant and can be ignored in the circuits. The proof is the same as that of Case\hspace{0.1cm}1.
\end{proof}

\begin{proposition}\label{prop-soundproperty}
	Let $\mathbf{A}, \mathbf{B}$ be two reversible circuits, $q$ a bit that does not occur in $\mathbf{A},\mathbf{B}$, and $\mathbf{A}',\mathbf{B}'$ obtained by adding $q$ as a positive control bit to all gates in $\mathbf{A},\mathbf{B}$. If $\mathbf{A} \Leftrightarrow \mathbf{B}$, then $\mathbf{A}' \Leftrightarrow \mathbf{B}'$.
\end{proposition}
\begin{proof}
	If $\mathbf{A} \Leftrightarrow \mathbf{B}$, then $\mathbf{A}$ can be transformed into $\mathbf{B}$ by the rules in $\mathcal{RC}$. We notice that all other rules are still valid if adding a new positive control bit to the gates except Rule~\ref{rule-swap-eliman}. So we only need to consider the proof for this rule.
	 
	The CNOT gate becomes a Toffoli gate if adding a positive control bit to it, which violates the requirement of Rule~\ref{rule-swap-eliman}. The main idea is that we use auxiliary gates generated by Rule~\ref{rule-gate-elimn}, in which $q$ is a negative control bit, to eliminate the positive control bit $q$ in the Toffoli gates by Rule~\ref{rule-qubit-add-remove}, and then obtain CNOT gates. Rule~\ref{rule-gateswap} can be used to move the gates in the transformation. Finally, Rule~\ref{rule-swap-eliman} can be applied.  The following is an example of the proof. Suppose that we have
	\[
	\scalebox{0.7}{
		\begin{quantikz}[row sep={0.7cm,between origins}]
			& \qw  	   & \qw 	   & \qw		& \octrl{1} & \qw \\
			& \ctrl{1} & \targ{}   & \ctrl{1}   & \targ{}   & \qw \\
			& \targ{}  & \ctrl{-1} & \targ{}    & \ctrl{-1} & \qw 
		\end{quantikz}}
		\Longleftrightarrow
    \scalebox{0.7}{
		\begin{quantikz}[row sep={0.7cm,between origins}]
			& \octrl{1}	& \qw  	   & \qw 	   & \qw	  & \qw \\
			& \ctrl{1}  & \ctrl{1} & \targ{}   & \ctrl{1} & \qw \\
			& \targ{}	& \targ{}  & \ctrl{-1} & \targ{}  & \qw 
	\end{quantikz}}
	\]
	and would like to get 
	\[
	\scalebox{0.7}{
	\begin{quantikz}[row sep={0.7cm,between origins}]
		& \ctrl{2} & \ctrl{2}  & \ctrl{2}	& \ctrl{1}  & \qw \\
		& \qw  	   & \qw 	   & \qw		& \octrl{1} & \qw \\
		& \ctrl{1} & \targ{}   & \ctrl{1}   & \targ{}   & \qw \\
		& \targ{}  & \ctrl{-1} & \targ{}    & \ctrl{-1} & \qw 
	\end{quantikz}}
	\Longleftrightarrow
    \scalebox{0.7}{
	\begin{quantikz}[row sep={0.7cm,between origins}]
		& \ctrl{1}  & \ctrl{2} & \ctrl{2}  & \ctrl{2} & \qw \\
		& \octrl{1}	& \qw  	   & \qw 	   & \qw	  & \qw \\
		& \ctrl{1}  & \ctrl{1} & \targ{}   & \ctrl{1} & \qw \\
		& \targ{}	& \targ{}  & \ctrl{-1} & \targ{}  & \qw 
	\end{quantikz}}
	\]
	We can do the transformation as follows.
	\[
	\scalebox{0.7}{
		\begin{quantikz}[row sep={0.7cm,between origins}]
			& \ctrl{2} & \ctrl{2}  & \ctrl{2}	& \ctrl{1}  & \qw \\
			& \qw  	   & \qw 	   & \qw		& \octrl{1} & \qw \\
			& \ctrl{1} & \targ{}   & \ctrl{1}   & \targ{}   & \qw \\
			& \targ{}  & \ctrl{-1} & \targ{}    & \ctrl{-1} & \qw 
		\end{quantikz}}
	\overset{\text{Rule~\ref{rule-gate-elimn}}}{\Longleftrightarrow}
	\]
	\[
	\scalebox{0.7}{
		\begin{quantikz}[row sep={0.7cm,between origins}]
			& \octrl{2}& \octrl{2}& \octrl{2}& \octrl{2}& \octrl{2}& \octrl{2} & \ctrl{2} & \ctrl{2}  & \ctrl{2}	& \ctrl{1}  & \qw \\
			& \qw  	   & \qw 	   & \qw	& \qw  	   & \qw 	   & \qw& \qw  	   & \qw 	   & \qw		& \octrl{1} & \qw \\
			& \ctrl{1} & \targ{}   & \ctrl{1}& \ctrl{1} & \targ{}   & \ctrl{1} & \ctrl{1} & \targ{}   & \ctrl{1}   & \targ{}   & \qw \\
			& \targ{}  & \ctrl{-1} & \targ{} & \targ{}  & \ctrl{-1} & \targ{} & \targ{}  & \ctrl{-1} & \targ{}    & \ctrl{-1} & \qw 
	\end{quantikz}}
	\overset{\text{Rule~\ref{rule-gateswap}}}{\Longleftrightarrow}
	\]
	\[
	\scalebox{0.7}{
		\begin{quantikz}[row sep={0.7cm,between origins}]
			& \octrl{2}& \octrl{2} & \octrl{2}  & \octrl{2} & \ctrl{2} & \octrl{2} & \ctrl{2}  & \octrl{2} & \ctrl{2}	& \ctrl{1}  & \qw \\
			& \qw  	   & \qw 	   & \qw		& \qw & \qw  	   & \qw & \qw 	   & \qw& \qw		& \octrl{1} & \qw \\
			& \ctrl{1} & \targ{}   & \ctrl{1}	& \ctrl{1}& \ctrl{1} & \targ{}& \targ{}   & \ctrl{1}& \ctrl{1}   & \targ{}   & \qw \\
			& \targ{}  & \ctrl{-1} & \targ{}	& \targ{}  & \targ{}  & \ctrl{-1} & \ctrl{-1} & \targ{}& \targ{}    & \ctrl{-1} & \qw 
	\end{quantikz}}
	\overset{\text{Rule~\ref{rule-qubit-add-remove}}}{\Longleftrightarrow}
	\]
	\[
	\scalebox{0.7}{
	\begin{quantikz}[row sep={0.7cm,between origins}]
		& \octrl{2}& \octrl{2}& \octrl{2} & \qw & \qw  & \qw	& \ctrl{1}  & \qw \\
		& \qw  	   & \qw 	   & \qw& \qw  	   & \qw 	   & \qw		& \octrl{1} & \qw \\
		& \ctrl{1} & \targ{}   & \ctrl{1} & \ctrl{1} & \targ{}   & \ctrl{1}   & \targ{}   & \qw \\
		& \targ{}  & \ctrl{-1} & \targ{} & \targ{}  & \ctrl{-1} & \targ{}    & \ctrl{-1} & \qw 
	\end{quantikz}}
	\overset{\text{Rule~\ref{rule-swap-eliman}}}{\Longleftrightarrow}
	\]
	\[
	\scalebox{0.7}{
		\begin{quantikz}[row sep={0.7cm,between origins}]
			& \octrl{2}& \octrl{2}& \octrl{2}  & \ctrl{1} & \qw & \qw  & \qw	  & \qw \\
			& \qw  	   & \qw 	   & \qw 	& \octrl{1} & \qw  	   & \qw 	   & \qw	 & \qw \\
			& \ctrl{1} & \targ{}   & \ctrl{1} & \ctrl{1} & \ctrl{1} & \targ{}   & \ctrl{1}     & \qw \\
			& \targ{}  & \ctrl{-1} & \targ{} & \targ{} & \targ{}  & \ctrl{-1} & \targ{}     & \qw 
	\end{quantikz}}
	\Longleftrightarrow
    \scalebox{0.7}{
	\begin{quantikz}[row sep={0.7cm,between origins}]
	& \ctrl{1}  & \ctrl{2} & \ctrl{2}  & \ctrl{2} & \qw \\
	& \octrl{1}	& \qw  	   & \qw 	   & \qw	  & \qw \\
	& \ctrl{1}  & \ctrl{1} & \targ{}   & \ctrl{1} & \qw \\
	& \targ{}	& \targ{}  & \ctrl{-1} & \targ{}  & \qw 
	\end{quantikz}}
	\]
The last transformation is obtained by applying Rule~\ref{rule-gate-elimn}, \ref{rule-qubit-add-remove}, and \ref{rule-gateswap} again.
\end{proof}

\subsection{Derived rules}
The following rules are all derivable from $\mathcal{RC}$ and are commonly employed in reversible circuit optimization. We utilize them to simplify the proof of completeness.

\begin{enumerate}[\textbf{Rule} 1.]
	\setcounter{enumi}{5}
	\item\label{rule-binary-generate} Let $A =\mathbf{G}[P,N,q]$, $Q=\{q_1,\dots,q_m\}$ a set of bits such that $Q\cap(P\cup N \cup \{q\}) =\emptyset$, and $\{Q_0,Q_1,\dots,Q_{2^m-1}\}$ the power set of $Q$. For each $0\leq i \leq {2^m-1}$, define 
	\[
	A_{Q_i}=\mathbf{G}[P\cup Q_i,N\cup Q/Q_i,q].
	\]
	Then	
	\[
	 A \equiv A_{Q_0}A_{Q_1}\cdots A_{Q_{2^m-1}}.
	\]
	For example,
	\[
	\scalebox{0.7}{
		\begin{quantikz}[row sep={0.7cm,between origins}]
			& \qw  	   & \qw \\
			& \qw  	   & \qw \\
			& \ctrl{1} & \qw \\
			& \targ{}  & \qw 
		\end{quantikz}}
		\equiv
		\scalebox{0.7}{
		\begin{quantikz}[row sep={0.7cm,between origins}]
			& \octrl{1}	& \octrl{1} & \ctrl{1}& \ctrl{1} & \qw \\
			& \octrl{1}	& \ctrl{1} & \octrl{1}& \ctrl{1} & \qw \\
			& \ctrl{1}  & \ctrl{1} & \ctrl{1}& \ctrl{1}	& \qw \\
			& \targ{}	& \targ{}  & \targ{} & \targ{}  & \qw 
	\end{quantikz}}
	\]
	
	\item\label{rule-gateswap-derived} If $A =\mathbf{G}[P_1,N_1,q]$, $B =\mathbf{G}[P_2,N_2,q]$ are two gates, then
	\[
	A B \equiv B A.
	\]
	For example,
	\[
	\scalebox{0.7}{
	\begin{quantikz}[row sep={0.7cm,between origins}]
		& \octrl{1}  & \qw		 &  \qw \\
		& \ctrl{2}   & \ctrl{2}  &  \qw \\
		& \wave 	 &           &  \qw \\
		& \ctrl{1}   & \octrl{1} &  \qw \\
		& \targ{}    & \targ{}   &  \qw
	\end{quantikz}}
	\equiv
	\scalebox{0.7}{
	\begin{quantikz}[row sep={0.7cm,between origins}]
		& \qw		 & \octrl{1} &  \qw \\
		& \ctrl{2}   & \ctrl{2}  &  \qw \\
		& \wave 	 &           &  \qw \\
		& \octrl{1}  & \ctrl{1}  &  \qw \\
		& \targ{}    & \targ{}   &  \qw
	\end{quantikz}}
	\]

	\item\label{rule-MPMCTtoMCT} If $A =\mathbf{G}[P,N,q]$, $B =\mathbf{G}[P\cup N,\emptyset,q]$, where the set $N=\{q_1,\dots,q_m\}$, then
	\[
	A \equiv \mathrm{X}[q_1] \cdots \mathrm{X}[q_m] B\, \mathrm{X}[q_1] \cdots \mathrm{X}[q_m].
	\]
	For example, 
	\[
	\scalebox{0.7}{
		\begin{quantikz}[row sep={0.7cm,between origins}]
			& \octrl{1}    & \qw \\
			& \ctrl{2}     & \qw \\
			& \wave  	  & \qw \\
			& \octrl{1}    & \qw \\
			& \targ{}      & \qw
		\end{quantikz}}
		\equiv
		\scalebox{0.7}{
		\begin{quantikz}[row sep={0.7cm,between origins}]
			& \gate{X}  & \ctrl{1}   & \gate{X} & \qw \\
			& \qw		& \ctrl{2}   & \qw 		& \qw \\
			& \qw 		& \wave  	 & \qw  	& \qw \\
			& \gate{X}  & \ctrl{1}   & \gate{X} & \qw \\
			& \qw  		& \targ{}	 & \qw 		& \qw
	\end{quantikz}}
	\]

	\item\label{rule-swap-gate-changectrl}  If $A =\mathbf{G}[P_1,N_1,p]$,  $B_1 =\mathbf{G}[P_2\cup\{p\},N_2,q]$, and $B_2 =\mathbf{G}[P_2,N_2\cup\{p\},q]$ are gates such that $P_1\subseteq P_2$ and $N_1\subseteq N_2$, then
		\[
		A B_1 \equiv B_2 A.
		\]
		For example,
		\[
		\scalebox{0.7}{
		\begin{quantikz}[row sep={0.7cm,between origins}]
			& \qw		 & \ctrl{1}  &  \qw \\
			& \ctrl{2}   & \ctrl{2}  &  \qw \\
			& \wave 	 &           &  \qw \\
			& \targ{}    & \ctrl{1}  &  \qw \\
			& \qw	     & \targ{}   &  \qw
		\end{quantikz}}
		\equiv
		\scalebox{0.7}{
		\begin{quantikz}[row sep={0.7cm,between origins}]
			& \ctrl{1}   & \qw		 &  \qw \\
			& \ctrl{2}   & \ctrl{2}  &  \qw \\
			& \wave 	 &           &  \qw \\
			& \octrl{1}  & \targ{}   &  \qw \\
			& \targ{}  	 & \qw    	 &  \qw
		\end{quantikz}}
		\]
	
	\item\label{rule-ABA-BAB} If $A =\mathbf{G}[P\cup\{p\},\emptyset,q]$, $B =\mathbf{G}[P\cup\{q\},\emptyset,p]$ are two gates such that $p \neq q$, then
	\[
	ABA \equiv BAB.
	\]
	For example,
	\[
	\scalebox{0.7}{
		\begin{quantikz}[row sep={0.7cm,between origins}]
			& \ctrl{2} &\ctrl{2}  & \ctrl{2} & \qw \\
			& \wave    &   		  &			 & \qw \\
			& \ctrl{1} &\targ{}   & \ctrl{1} & \qw \\
			& \targ{}  &\ctrl{-1} & \targ{}  & \qw 
		\end{quantikz}}
		\equiv
		\scalebox{0.7}{
		\begin{quantikz}[row sep={0.7cm,between origins}]
			&\ctrl{2} & \ctrl{2} &\ctrl{2}  & \qw \\
			&\wave   &   		 &			& \qw \\
			&\targ{}  & \ctrl{1} &\targ{}   & \qw \\
			&\ctrl{-1}& \targ{}  &\ctrl{-1} & \qw 
		\end{quantikz}}
		\]
\end{enumerate}

Rule~\ref{rule-binary-generate} can be easily derived from Rule~\ref{rule-qubit-add-remove}. For Rule~\ref{rule-gateswap-derived}, if $A,B$ have a common control bit that has different polarities in the two gates, respectively, then Rule~\ref{rule-gateswap} can be applied directly. 
Otherwise, we first use Rule~\ref{rule-binary-generate} to expand the two gates $A,B$ to a sequence of gates such that they have the same control bits, then use Rule~\ref{rule-gateswap} to move these gates, and next use Rule~\ref{rule-binary-generate} again to reduce these gates, as shown in the following example.
\[
\scalebox{0.7}{
	\begin{quantikz}[row sep={0.7cm,between origins}]
		&\qw & \ctrl{1}    & \qw \\
		&\qw & \ctrl{1}    & \qw \\
		& \ctrl{1} & \ctrl{1}   & \qw \\
		& \targ{}& \targ{}    & \qw 
	\end{quantikz}}
	\Longleftrightarrow
	\scalebox{0.7}{
	\begin{quantikz}[row sep={0.7cm,between origins}]
	&\ctrl{1} &\octrl{1} &\ctrl{1} & \octrl{1} & \ctrl{1}    & \qw \\
	&\ctrl{1} &\ctrl{1} &\octrl{1} & \octrl{1} & \ctrl{1}    & \qw \\
	&\ctrl{1} &\ctrl{1} & \ctrl{1} & \ctrl{1} & \ctrl{1}   & \qw \\
	& \targ{} & \targ{} & \targ{} & \targ{}& \targ{}    & \qw 
	\end{quantikz}}
	\Longleftrightarrow
	\scalebox{0.7}{
	\begin{quantikz}[row sep={0.7cm,between origins}]
	& \ctrl{1}  &\qw   & \qw \\
	& \ctrl{1}   &\qw & \qw \\
	& \ctrl{1} & \ctrl{1}   & \qw \\
	& \targ{} & \targ{}    & \qw 
	\end{quantikz}
}
\]

Rule~\ref{rule-MPMCTtoMCT} decomposes an MPMCT gate to a combination of an MCT gate and X gates, where all negative control bits become positive control bits by adding X gates before and after. To derive the rule, we first use Rule~\ref{rule-qubit-add-remove} to expand the X gate, then apply Rule~\ref{rule-completeness}, and next use Rule~\ref{rule-gateswap} and \ref{rule-gate-elimn} to move and delete gates, respectively, as shown in the following example.
\[
\scalebox{0.7}{
	\begin{quantikz}[row sep={0.7cm,between origins}]
		& \qw & \ctrl{1}  & \qw & \qw \\
		& \gate{X} & \ctrl{1}  & \gate{X}& \qw \\
		& \qw & \targ{}  & \qw & \qw 
	\end{quantikz}}
	\overset{\text{Rule~\ref{rule-qubit-add-remove}}}{\Longleftrightarrow}
	\scalebox{0.7}{
	\begin{quantikz}[row sep={0.7cm,between origins}]
	& \octrl{1} & \ctrl{1} & \ctrl{1} & \ctrl{1} & \octrl{1}  & \qw \\
	& \targ{} & \targ{} & \ctrl{1} 	& \targ{} & \targ{}  & \qw \\
	& \qw & \qw & \targ{}  & \qw & \qw  & \qw 
	\end{quantikz}}
	\overset{\text{Rule~\ref{rule-qubit-add-remove}}}{\Longleftrightarrow}
\]
\[
\scalebox{0.7}{
	\begin{quantikz}[row sep={0.7cm,between origins}]
		& \octrl{1} & \ctrl{1} & \ctrl{1} & \ctrl{1} & \ctrl{1} & \ctrl{1} & \octrl{1}  & \qw \\
		& \targ{} & \targ{} & \targ{} & \ctrl{1} 	& \targ{} & \targ{} & \targ{}  & \qw \\
		& \qw & \octrl{-1} & \ctrl{-1} & \targ{} & \ctrl{-1} & \octrl{-1} & \qw & \qw
	\end{quantikz}}
	\overset{\text{Rule~\ref{rule-completeness}}}{\Longleftrightarrow}
	\scalebox{0.7}{
	\begin{quantikz}[row sep={0.7cm,between origins}]
 	& \octrl{1} & \ctrl{1} & \octrl{1}   & \qw \\
	& \targ{} & \octrl{1} 	& \targ{}  & \qw \\
	& \qw & \targ{}  & \qw & \qw  
	\end{quantikz}}
\]
\[
\overset{\text{Rule~\ref{rule-gateswap}}}{\Longleftrightarrow}
\scalebox{0.7}{
	\begin{quantikz}[row sep={0.7cm,between origins}]
		& \ctrl{1} & \octrl{1} & \octrl{1}   & \qw \\
		& \octrl{1} & \targ{} 	& \targ{}  & \qw \\
		& \targ{}  & \qw & \qw   & \qw 
	\end{quantikz}}
	\overset{\text{Rule~\ref{rule-gate-elimn}}}{\Longleftrightarrow}
	\scalebox{0.7}{
	\begin{quantikz}[row sep={0.7cm,between origins}]
	& \ctrl{1}    & \qw \\
	& \octrl{1}   & \qw \\
	& \targ{}     & \qw 
	\end{quantikz}}
\]

Rule~\ref{rule-swap-gate-changectrl} can be derived from Rule~\ref{rule-gate-elimn}, \ref{rule-gateswap}, \ref{rule-binary-generate}, and \ref{rule-MPMCTtoMCT}. This rule is used to change the polarity of a control bit by an X gate in the proof of completeness. For example,
\[
	\scalebox{0.7}{
	\begin{quantikz}[row sep={0.7cm,between origins}]
		&\qw & \ctrl{1} &\ctrl{1}  & \ctrl{1} & \qw \\
		&\gate{X}	& \octrl{2} &\octrl{2}  & \octrl{2} & \qw \\
		& \wave    &   	&	  &			 & \qw \\
		&\qw	& \ctrl{1} &\targ{}   & \ctrl{1} & \qw \\
		&\qw	& \targ{}  &\octrl{-1} & \targ{}  & \qw 
	\end{quantikz}}
	\Longleftrightarrow
	\scalebox{0.7}{
	\begin{quantikz}[row sep={0.7cm,between origins}]
		& \ctrl{1} &\ctrl{1}  & \ctrl{1} &\qw & \qw \\
		& \ctrl{2} & \ctrl{2} &\ctrl{2}  & \gate{X} & \qw \\
		& \wave  &   		 &			& \qw &\qw\\
		& \ctrl{1}&\targ{}  & \ctrl{1}  &\qw& \qw \\
		& \targ{} &  \octrl{-1} &\targ{} & \qw & \qw 
\end{quantikz}}
\]

For Rule~\ref{rule-ABA-BAB}, first by Rule~\ref{rule-gate-elimn} and \ref{rule-swap-eliman} we can conclude that
\[
\scalebox{0.7}{
	\begin{quantikz}[row sep={0.7cm,between origins}]
		& \ctrl{1} &\targ{}   & \ctrl{1} & \qw \\
		& \targ{}  &\ctrl{-1} & \targ{}  & \qw 
	\end{quantikz}}
	\overset{\text{Rule~\ref{rule-gate-elimn}}}{\Longleftrightarrow}
	\scalebox{0.7}{
	\begin{quantikz}[row sep={0.7cm,between origins}]
		& \ctrl{1} &\targ{}   & \ctrl{1} &\targ{}   &\targ{} & \qw \\
		& \targ{}  &\ctrl{-1} & \targ{}  &\ctrl{-1} &\ctrl{-1} & \qw 
	\end{quantikz}}
\]
\[
\overset{\text{Rule~\ref{rule-swap-eliman}}}{\Longleftrightarrow}
	\scalebox{0.7}{
	\begin{quantikz}[row sep={0.7cm,between origins}]
		& \ctrl{1} & \ctrl{1} &\targ{}   & \ctrl{1} &\targ{}   & \qw \\
		& \targ{}  & \targ{}  &\ctrl{-1} & \targ{}  &\ctrl{-1} & \qw 
	\end{quantikz}}
	\overset{\text{Rule~\ref{rule-gate-elimn}}}{\Longleftrightarrow}
	\scalebox{0.7}{
	\begin{quantikz}[row sep={0.7cm,between origins}]
		&\targ{}  & \ctrl{1} &\targ{}   & \qw \\
		&\ctrl{-1}& \targ{}  &\ctrl{-1} & \qw 
	\end{quantikz}}
\]
Then by Proposition~\ref{prop-soundproperty} we obtain Rule~\ref{rule-ABA-BAB}.

In the template-based circuit optimization method, a pair of circuits $(\mathbf{A},\mathbf{B})$ forms a template, where $\mathbf{A}$ typically has a larger size than $\mathbf{B}$. During circuit optimization, a subcircuit $\mathbf{A}$ within a circuit $\mathbf{C}$ can be replaced by $\mathbf{B}$ to reduce the cost. To ensure the optimized circuit remains functionally equivalent to $\mathbf{C}$, the two circuits $\mathbf{A}$ and $\mathbf{B}$ must be equivalent. By Theorem~\ref{thm-soundness}, the correctness of a template can be verified by checking whether $\mathbf{A}$ and $\mathbf{B}$ can be transformed into each other according to the rules in $\mathcal{RC}$. Furthermore, $\mathcal{RC}$ can be utilized to derive new templates.

\begin{example}\label{exa-template}
	The correctness of template~3.2 in Figure~4 of~\cite{Miller2003transformation} can be verified by the following transformation.
	\[
	\scalebox{0.75}{\begin{quantikz}[row sep={0.7cm,between origins}]
			& \qw	   & \ctrl{2} & \ctrl{1} & \qw \\
			& \ctrl{1} & \qw  	  & \targ{}  & \qw \\
			& \targ{}  & \targ{}  & \qw		 & \qw
	\end{quantikz}}
	\overset{\text{Rule~\ref{rule-qubit-add-remove}}}{\Longleftrightarrow}
	\scalebox{0.75}{\begin{quantikz}[row sep={0.7cm,between origins}]
			& \octrl{1}	 & \ctrl{1}	 & \ctrl{1}  & \ctrl{1}	& \ctrl{1} & \qw \\
			& \ctrl{1} 	 & \ctrl{1}	 & \octrl{1} & \ctrl{1}	& \targ{}  & \qw \\
			& \targ{} 	 & \targ{}   & \targ{}   & \targ{}	& \qw	   & \qw
	\end{quantikz}}
	\overset{\text{Rule~\ref{rule-gateswap}}}{\Longleftrightarrow}
	\]
	\[
	\scalebox{0.75}{\begin{quantikz}[row sep={0.7cm,between origins}]
			& \octrl{1}	 & \ctrl{1}	 & \ctrl{1}  & \ctrl{1}	& \ctrl{1} & \qw \\
			& \ctrl{1} 	 & \octrl{1} & \ctrl{1}  & \ctrl{1}	& \targ{}  & \qw \\
			& \targ{} 	 & \targ{}   & \targ{}   & \targ{}	& \qw	   & \qw
	\end{quantikz}}
	\overset{\text{Rule~\ref{rule-gate-elimn}}}{\Longleftrightarrow}
	\scalebox{0.75}{\begin{quantikz}[row sep={0.7cm,between origins}]
			& \octrl{1}	 & \ctrl{1}	 & \ctrl{1} & \qw \\
			& \ctrl{1} 	 & \octrl{1} & \targ{}  & \qw \\
			& \targ{} 	 & \targ{}   & \qw	   & \qw
	\end{quantikz}}
	\overset{\text{Rule~\ref{rule-swap-gate-changectrl}}}{\Longleftrightarrow}
	\]
	\[
	\scalebox{0.75}{\begin{quantikz}[row sep={0.7cm,between origins}]
			& \octrl{1}	 & \ctrl{1}	& \ctrl{1} & \qw \\
			& \ctrl{1} 	 & \targ{} 	& \ctrl{1} & \qw \\
			& \targ{} 	 & \qw		& \targ{}  & \qw
	\end{quantikz}}
	\overset{\text{Rule~\ref{rule-gateswap}}}{\Longleftrightarrow}
	\scalebox{0.75}{\begin{quantikz}[row sep={0.7cm,between origins}]
			& \ctrl{1} & \octrl{1}	& \ctrl{1} & \qw \\
			& \targ{}  & \ctrl{1} 	& \ctrl{1} & \qw \\
			& \qw	  & \targ{} 	& \targ{}  & \qw
	\end{quantikz}}
	\overset{\text{Rule~\ref{rule-qubit-add-remove}}}{\Longleftrightarrow}
	\scalebox{0.75}{\begin{quantikz}[row sep={0.7cm,between origins}]
			& \ctrl{1} & \qw	  & \qw \\
			& \targ{}  & \ctrl{1} & \qw \\
			& \qw	   & \targ{}  & \qw
	\end{quantikz}}
	\]
\end{example}

\section{Completeness}\label{sec-completeness}
In this section, we show that any two equivalent reversible circuits can be transformed into each other by applying the rules in $\mathcal{RC}$. The proof is based on the canonical forms of reversible circuits, where every reversible circuit has a unique canonical form.

\subsection{Canonical form}
The $n$-hypercube graph is an undirected graph defined on the set $\{0,1\}^n$ such that there is an edge between two nodes $a,b$ iff they only differ in exactly one bit. All hypercube graphs are Hamiltonian~\cite{Skiena1991implementing}. Let 
\[
\mathbb{H}=(a_0,a_1,\dots,a_{2^n-1})
\]
be a Hamiltonian path of an $n$-hypercube graph. Every $n$-ary reversible function $f$ defines a permutation
\[
\begin{pmatrix}
	a_0 & a_1 & \dots & a_{2^n-2} & a_{2^n-1} \\
	b_0 & b_1 & \dots & b_{2^n-2} & b_{2^n-1} \\
\end{pmatrix}
\]
such that $f(a_i)=b_i$ ($0\leq i < 2^n$). To simplify notation, we use $(b_0,b_1,\dots,b_{2^n-1})_{\mathbb{H}}$ to denote the permutation whose first row is given by $\mathbb{H}$.
Define
\[
\Delta_\mathbb{H}=\{M_0,M_1,\dots,M_{2^n-2}\}
\]
to be the set of $n$-bit MPMCT gates where for each $0\leq i \leq 2^n-2$, the gate $M_i$ exchanges $a_i$ and $a_{i+1}$, namely, the polarities of the control bits of $M_i$ coincide with the values of the common bits of $a_i$ and $a_{i+1}$, i.e., $q_j$ ($1\leq j \leq n$) is a positive (resp. negative) control bit of $M_i$ iff the $j$-th bit of both $a_i$ and $a_{i+1}$ is 1 (resp. 0). It is easily seen that $M_i$ defines the permutation 
\[
(a_0,\dots,a_{i-1},a_{i+1},a_{i},a_{i+2},\cdots,a_{2^n-1})_{\mathbb{H}}.
\]

Let $\mathbf{C}=M_i M_{i+1} \cdots M_{i+j-1}$ be a sequence of consecutive gates from $\Delta_\mathbb{H}$ ($0\leq i\leq 2^n-2$, $j\geq 1$). By the definition of $\Delta_\mathbb{H}$, the circuit $\mathbf{C}$ defines the following permutation
\[
(a_0,\dots,a_{i-1},\underbrace{a_{i+j},a_{i},\dots,a_{i+j-1},}_{\text{cyclic shift by 1 position}} a_{i+j+1},\dots,a_{2^n-1})_{\mathbb{H}},
\]
which maps $a_i$ to $a_{i+j}$, and $a_k$ to $a_{k-1}$ ($i+1\leq k \leq i+j$). Therefore, if a reversible circuit $\mathbf{C}'$ defines a permutation $(b_0,b_1,\dots,b_{2^n-1})_{\mathbb{H}}$, then the circuit $\mathbf{C}\mathbf{C}'$ defines the permutation
\[
(b_0,\dots,b_{i-1},\underbrace{b_{i+j},b_{i},\dots,b_{i+j-1},}_{\text{cyclic shift by 1 position}} b_{i+j+1},\dots,b_{2^n-1})_{\mathbb{H}}.
\]

\begin{definition}[\textbf{Canonical form}]\label{def-canonicalform}
	An $n$-bit reversible circuit is in the canonical form based on $\mathbb{H}$ if it has the form $\mathbf{C}_m\mathbf{C}_{m-1}\cdots\mathbf{C}_1\mathbf{C}_0$ such that
	\begin{enumerate}[(1)]
		\item for each $0\leq i \leq m$, the subcircuit
		\[
		\mathbf{C}_i= M_x M_{x+1} M_{x+2} \cdots M_{x+k}
		\] 
		is a sequence of consecutive gates from $\Delta_\mathbb{H}$ ($0\leq x \leq x+k\leq 2^n-2$),
		\item for $\mathbf{C}_i= M_x \cdots M_{x+k}$ and $\mathbf{C}_j= M_y \cdots M_{y+l}$, if $i<j$, then $x < y$.
	\end{enumerate}
\end{definition}

By the above definition, we immediately see the following fact.

\begin{fact}\label{fact-canonicalform}
	Let $\mathbf{C}_m\mathbf{C}_{m-1}\cdots\mathbf{C}_1\mathbf{C}_0$ be an arbitrary reversible circuit in the canonical form based on $\mathbb{H}$. Then
	\begin{enumerate}[(1)]
		\item\label{fact-circuitnumber} $m\leq 2^n-2$; 
		\item\label{fact-firstgate} the first gate $M_x$ of $\mathbf{C}_i$ does not occur in $\mathbf{C}_m\cdots \mathbf{C}_{i+1}$, and for every gate $M_z$ in $\mathbf{C}_m\cdots \mathbf{C}_{i+1}$, $z>x$;
		\item\label{fact-gatenumber} the gate $M_i$ ($0\leq i \leq 2^n-2$) occurs at most $i+1$ times in the canonical form.
	\end{enumerate}
\end{fact}

\begin{remark}
	The choice of the Hamiltonian path $\mathbb{H}$ is arbitrary; it has no influence on the proof of completeness. Since every $n$-bit reversible circuit computes a reversible function, and every $n$-ary reversible function defines a permutation on $\{0,1\}^n$. So if $\mathbb{H}$ is given, we can construct a unique element moving process on $\mathbb{H}$ to get the permutation. Each moving step is realized by an $n$-bit MPMCT gate.
\end{remark}

The gates in $\Delta_\mathbb{H}$ are universal, as is shown in the following proposition.

\begin{proposition}[\textbf{Universality}]
	Every $n$-ary reversible function can be computed by a unique $n$-bit reversible circuit that is in the canonical form based on $\mathbb{H}$.
\end{proposition}
\begin{proof}
	Let $f$ be an arbitrary $n$-ary reversible function that defines a permutation $(a_{x_0},a_{x_1},\dots,a_{x_{2^n-1}})_{\mathbb{H}}$. We construct a reversible circuit that computes $f$ in the canonical form based on $\mathbb{H}$.
	
	Set $\mathbf{C}_0= M_0 M_1 \cdots M_{x_0-1}$. It defines the permutation 
	\[
	(a_{x_0},b_1,b_2,\dots,b_{2^n-1})_{\mathbb{H}},
	\]
	where $\{b_1,b_2,\dots,b_{2^n-1}\}=\{a_{x_1},\dots,a_{x_{2^n-1}}\}$. If $a_{x_1}=b_j$ ($1\leq j \leq 2^n-1$), then set $\mathbf{C}_1= M_1 M_2 \cdots M_{j-1}$. Thus, $\mathbf{C}_1\mathbf{C}_0$ defines the permutation
	\[
	(a_{x_0},a_{x_1},c_2,c_3,\dots,c_{2^n-1})_{\mathbb{H}},
	\]
	where $\{c_2,c_3,\dots,c_{2^n-1}\} = \{a_{x_2},\dots,a_{x_{2^n-1}}\}$. 
	
	Suppose that $\mathbf{C}_i\cdots\mathbf{C}_1 \mathbf{C}_0$ defines the permutation 
	\[
	(a_{x_0},a_{x_1},\dots,a_{x_i},d_{i+1},\dots,d_{2^n-1})_{\mathbb{H}},
	\]
	and $a_{x_{i+1}}=d_l$ ($i+1\leq l \leq 2^n-1$). Set $\mathbf{C}_{i+1}= M_{i+1} M_{i+2} \cdots M_{l-1}$. The reversible circuit $\mathbf{C}_{i+1}\mathbf{C}_i\cdots\mathbf{C}_1 \mathbf{C}_0$ defines the permutation
	\[
	(a_{x_0},a_{x_1},\dots,a_{x_i},a_{x_{i+1}},e_{i+2},\dots,e_{2^n-1})_{\mathbb{H}}.
	\]
	Repeated application of the process can finally generate a reversible circuit $\mathbf{C}=\mathbf{C}_m\mathbf{C}_{m-1}\cdots\mathbf{C}_1\mathbf{C}_0$ that moves every $a_{x_i}$ ($0\leq i \leq 2^n-1$) to its position in the permutation defined by $f$. It is easy to check that $\mathbf{C}$ is unique and is in the canonical form based on $\mathbb{H}$ from the construction above.
	
		\begin{algorithm}[t]
		\caption{\label{algo-contruction}Canonical circuit construction from a reversible function}
		\KwIn{A Hamiltonian path $\mathbb{H} = (a_0, a_1, \ldots, a_{2^n-1})$ of the $n$-hypercube and the permutation $(a_{x_0}, a_{x_1}, \ldots, a_{x_{2^n-1}})_\mathbb{H}$ defined by a reversible function $f$.}
		\KwOut{Canonical circuit $\mathbf{C}=\mathbf{C}_m\mathbf{C}_{m-1}\cdots\mathbf{C}_1\mathbf{C}_0$.}
		\SetAlgoLined
		\PrintSemicolon
		
		$\mathbf{C} \leftarrow \emptyset$\;
		$current \leftarrow (a_0, a_1, \ldots, a_{2^n-1})_\mathbb{H}$\;
		
		\For{$i \leftarrow 0$ \KwTo $2^n-1$}{
			$l \leftarrow$ index such that $current[l] = a_{x_i}$\;
	
			\If{$l > i$}{
				$\mathbf{C}_i \leftarrow \emptyset$\;
				\For{$j \leftarrow i$ \KwTo $l-1$}{
					
					$p \leftarrow$ the bit where $a_j$ and $a_{j+1}$ differ\;
					
					// 1 and 0 represent positive and negative control bits respectively //
					$P=\{ a_j[m] = 1 \mid m \in [0,n-1] \setminus \{p\}\}$;
					$N=\{ a_j[m] = 0 \mid m \in [0,n-1] \setminus \{p\}\}$;
					$gate \leftarrow \mathbf{G}[P, N, p]$\;
					$\mathbf{C}_i \leftarrow \mathbf{C}_i \cdot gate$\;
				}
				$current \leftarrow (current[0], \dots, current[i-1]$, $a_{x_i}, current[i], \dots, current[l-1]$, $current[l+1], \dots, current[2^n-1])$\;		
				$\mathbf{C} \leftarrow \mathbf{C}_i \mathbf{C}$ \;
			}
		}
		\Return $\mathbf{C}$\;
	\end{algorithm}
	We give an algorithm that constructs the canonical circuit from a reversible function in Algorithm~\ref{algo-contruction}. The time complexity is $O(n \cdot 4^n)$ where $n$ is the number of bits in the circuit.
\end{proof}

\begin{example}\label{exam-Hamiltonian}
Given a Hamiltonian path
	\[
	\mathbb{H}=(000,001,011,010,110,111,101,100)
	\]
	of a 3-hypercube graph, let $s_0 s_1 s_2$ be an entry of $\mathbb{H}$, we assume that $s_i$ is the input of bit $q_i$ ($0\leq i \leq 2$). 
	The gates in the set $\Delta_\mathbb{H}=\{M_0,M_1,\dots,M_6\}$ are listed in Fig.~\ref{fig-hamiltonian}.
	\begin{figure}[t]
		\centering
		\scalebox{0.75}{\begin{quantikz}[row sep={0.7cm,between origins}]
				\setwiretype{n}	& \push{M_0}& \push{M_1} & \push{M_2} & \push{M_3}	& \push{M_4} & \push{M_5}	& \push{M_6} & \\
				\lstick{$q_0$:}	& \octrl{1} & \octrl{1}  & \octrl{1}  & \targ{}		& \ctrl{1} 	 & \ctrl{1}		& \ctrl{1} 	 & \qw \\
				\lstick{$q_1$:}	& \octrl{1} & \targ{}  	 & \ctrl{1}   & \ctrl{-1}	& \ctrl{1}   & \targ{}		& \octrl{1}  & \qw \\
				\lstick{$q_2$:}	& \targ{}   & \ctrl{-1}  & \targ{}	  & \octrl{-1}	& \targ{}    & \ctrl{-1}	& \targ{}    & \qw
		\end{quantikz}}
		\caption{\label{fig-hamiltonian} The gates in the set $\Delta_\mathbb{H}$.}
	\end{figure}
	
	\begin{table}[t]
		\centering
		\caption{\label{tab-truthtable} The truth table of the function $f$.}
		\begin{tabular}{ccc|ccc}
			\hline
			\multicolumn{3}{c|}{Input} & \multicolumn{3}{c}{Output} \\ \hline
			$q_0$   & $q_1$  & $q_2$  & $q_0$   & $q_1$   & $q_2$  \\ \hline
			0       & 0      & 0      &    0    &    1    &    0   \\
			0       & 0      & 1      &    0    &    0    &    0   \\
			0       & 1      & 0      &    1    &    0    &    1   \\
			0       & 1      & 1      &    0    &    0    &    1   \\
			1       & 0      & 0      &    1    &    0    &    0   \\
			1       & 0      & 1      &    1    &    1    &    1   \\
			1       & 1      & 0      &    0    &    1    &    1   \\
			1       & 1      & 1      &    1    &    1    &    0  \\ \hline
		\end{tabular}
	\end{table}
	Let $f$ be a reversible function whose truth table is given in Table~\ref{tab-truthtable}.  It is clear that $f$ defines the permutation 
	\[
	\begin{pmatrix}
		000 & 001 & 011 & 010 & 110 & 111 & 101 & 100 \\
		010 & 000 & 001 & 101 & 011 & 110 & 111 & 100 \\
	\end{pmatrix},
	\]
	denoted by $P_f= (010, 000, 001, 101, 011, 110, 111, 100)_{\mathbb{H}}$.
	We construct the canonical circuit that computes $f$ based on $\mathbb{H}$ by starting from the initial permutation $P_0$ that is defined by the empty circuit.
	\[
	\setlength\arraycolsep{2pt}
	\begin{matrix}
		& 0 & 1 & 2 & 3 & 4 & 5 & 6 & 7 \\
	P_0 =& (000, & 001, & 011, & 010, & 110, & 111, & 101, & 100)_{\mathbb{H}} \\
	\end{matrix}
	\]
	The first different element of $P_0$ and $P_f$ is at the 0th position. Hence, we first move 010 to the 0th position. This step can be decomposed into a sequence of operations that swap adjacent elements as follows.
	\begin{enumerate}
		\item Apply the gate $M_2$ on $P_0$ that exchanges 011 and 010. The circuit $M_2$ defines the permutation $P_1$.
		\[
		\setlength\arraycolsep{2pt}
		\begin{matrix}
			& 0 & 1 & 2 & 3 & 4 & 5 & 6 & 7 \\
			P_1 =& (000, & 001, & \textcolor{red}{010}, & 011, & 110, & 111, & 101, & 100)_{\mathbb{H}} \\
		\end{matrix}
		\]
		\item Apply gate $M_1$ on $P_1$ to move 010 to the 1st position. The circuit $M_1 M_2$ defines the permutation $P_2$.
		\[
		\setlength\arraycolsep{2pt}
		\begin{matrix}
			& 0 & 1 & 2 & 3 & 4 & 5 & 6 & 7 \\
			P_2 =& (000, & \textcolor{red}{010}, & 001, & 011, & 110, & 111, & 101, & 100)_{\mathbb{H}} \\
		\end{matrix}
		\]
		\item Apply gate $M_0$ on $P_2$ to move 010 to the 0th position. The circuit $M_0 M_1 M_2$ defines the permutation $P_3$.
		\[
		\setlength\arraycolsep{2pt}
		\begin{matrix}
			& 0 & 1 & 2 & 3 & 4 & 5 & 6 & 7 \\
			P_3 =& (\textcolor{red}{010}, & 000, & 001, & 011, & 110, & 111, & 101, & 100)_{\mathbb{H}} \\
		\end{matrix}
		\]
	\end{enumerate}

	Set $\mathbf{C}_0 = M_0 M_1 M_2$. The first different element of $P_3$ and $P_f$ is at the 3rd position. Then we move 101 to the 3rd position by the following steps.
	\begin{enumerate}
	\item Apply gate $M_5$ on $P_3$ to move 101 to the 5th position. The circuit $M_5 \mathbf{C}_0$ defines the permutation $P_4$.
	\[
	\setlength\arraycolsep{2pt}
	\begin{matrix}
		& 0 & 1 & 2 & 3 & 4 & 5 & 6 & 7 \\
		P_4 =& (010, & 000, & 001, & 011, & 110, & \textcolor{red}{101}, & 111, & 100)_{\mathbb{H}} \\
	\end{matrix}
	\]
	\item Apply gate $M_4$ on $P_4$ to move 101 to the 4th position. The circuit $M_4 M_5 \mathbf{C}_0$ defines the permutation $P_5$.
	\[
	\setlength\arraycolsep{2pt}
	\begin{matrix}
	& 0 & 1 & 2 & 3 & 4 & 5 & 6 & 7 \\
	P_5 =& (010, & 000, & 001, & 011, & \textcolor{red}{101}, & 110, & 111, & 100)_{\mathbb{H}} \\
	\end{matrix}
	\]
	\item Apply gate $M_3$ on $P_5$ to move 101 to the 3rd position. The circuit $M_3 M_4 M_5 \mathbf{C}_0$ defines the permutation $P_6$.
	\[
	\setlength\arraycolsep{2pt}
	\begin{matrix}
		& 0 & 1 & 2 & 3 & 4 & 5 & 6 & 7 \\
		P_6 =& (010, & 000, & 001, & \textcolor{red}{101}, & 011, & 110, & 111, & 100)_{\mathbb{H}} \\
	\end{matrix}
	\]
	\end{enumerate}
Set $\mathbf{C}_1= M_3 M_4 M_5$. The permutation $P_6$ equals $P_f$. Thus, the circuit $\mathbf{C}_1 \mathbf{C}_0= M_3 M_4 M_5 M_0 M_1 M_2$ is the canonical circuit that computes $f$ based on $\mathbb{H}$.
\[\mathbf{C}_1 \mathbf{C}_0=
\scalebox{0.75}{\begin{quantikz}[row sep={0.7cm,between origins}]
\setwiretype{n}	
		& \push{M_3}	& \push{M_4} & \push{M_5} 	& \push{M_0}& \push{M_1} & \push{M_2}  & \\
		& \targ{}		& \ctrl{1} 	 & \ctrl{1}		& \octrl{1} & \octrl{1}  & \octrl{1}   & \qw \\
		& \ctrl{-1}		& \ctrl{1}   & \targ{}		& \octrl{1} & \targ{}  	 & \ctrl{1}    & \qw \\
		& \octrl{-1}	& \targ{}    & \ctrl{-1} 	& \targ{}   & \ctrl{-1}  & \targ{}	   & \qw
\end{quantikz}}
\]
\end{example}

We next show that every reversible circuit that only consists of the gates in $\Delta_\mathbb{H}$ can be transformed into its unique canonical form. The proof is divided into a sequence of lemmas.

\begin{lemma}\label{lem-MMcommute}
	Let $(b_1,b_2,\dots,b_m)$ be a sequence of distinct strings from $\{0,1\}^n$ such that $b_i,b_{i+1}$ differ in exactly one bit, and $A_i$ an $n$-bit gate that exchanges $b_i,b_{i+1}$ ($1\leq i <m$).
	Then for any two gates $A_i, A_j$ ($1\leq i,j < m$), if $|i-j|\geq 2$, there must exist a control bit that has different polarities in $A_i$ and $A_j$, respectively. 
\end{lemma}
\begin{proof}
	Let the target bits of $A_i$ and $A_j$ be $q_l$ and $q_k$ ($1\leq l, k \leq n$), respectively. So the common control bits of $A_i$ and $A_j$ are $Q=\{q_1,\dots,q_n\}/\{q_l,q_k\}$.
	Suppose that $A_i$ exchanges $b_i$ and $b_{i+1}$, and $A_j$ exchanges $b_j$ and $b_{j+1} $. From $|i-j|\geq 2$ we know that $b_i$, $b_{i+1}$, $b_j$, and $b_{j+1}$ must be different from each other.
	
	Assume that for all $q\in Q$, the polarity of $q$ in $A_i$ is the same as that in $A_j$. This implies that $b_i$, $b_{i+1}$, $b_j$, $b_{j+1}$ have the same value in their $d$-th bit for each $d\in \{1,\dots,n\}/\{l,k\}$. Hence, whatever the $l$-th bits of $b_i,b_{i+1}$ and the $k$-th bits of $b_j,b_{j+1}$ are, there always be a string from $b_i, b_{i+1}$ that equals a string from $b_j, b_{j+1}$, a contradiction. Therefore, there must be a $q\in Q$ such that $q$ has different polarities in $A_i$ and $A_j$, respectively. 
\end{proof}

By Lemma~\ref{lem-MMcommute}, it is clear that for any $M_i, M_j\in \Delta_\mathbb{H}$, if $|i-j|\geq 2$, then there is a control bit that has different polarities in $M_i$ and $M_j$, respectively. By Rule~\ref{rule-gateswap} we have $M_i M_j\Leftrightarrow M_j M_i$.

\begin{lemma}\label{lem-ABA-BAB}
	Let	$A =\mathbf{G}[P_1,N_1,p]$, $B =\mathbf{G}[P_2,N_2,q]$ be two gates satisfying the following conditions:
	\begin{itemize}
		\item $P_1\cup N_1 \cup \{p\} = P_2\cup N_2 \cup \{q\}$,
		\item if $p'$ is a common control bit of $A$ and $B$, then the polarity of $p'$ in $A$ is the same as that in $B$.
	\end{itemize}
	Then
	\[ABA\Leftrightarrow BAB.\]
\end{lemma}
\begin{proof}
	The lemma follows easily by Rule~\ref{rule-gate-elimn}, \ref{rule-gateswap-derived}, \ref{rule-swap-gate-changectrl}, and \ref{rule-ABA-BAB}. We use the X gate to change the negative control bits to positive control bits, and apply Rule~\ref{rule-ABA-BAB}, and then change the positive control bits back to negative control bits. For example,
	\[
		\scalebox{0.7}{
		\begin{quantikz}[row sep={0.7cm,between origins}]
			\setwiretype{n}	& \push{A}& \push{B}& \push{A}  & \\
			& \ctrl{1} &\ctrl{1}  & \ctrl{1} & \qw \\
			& \octrl{1} &\octrl{1}  & \octrl{1} & \qw \\
			\lstick{$q$:} & \ctrl{1} &\targ{}   & \ctrl{1} & \qw \\
			\lstick{$p$:} & \targ{}  &\octrl{-1} & \targ{}  & \qw 
		\end{quantikz}}
		\overset{\text{Rule~\ref{rule-gate-elimn}}}{\Longleftrightarrow}
		\scalebox{0.7}{
		\begin{quantikz}[row sep={0.7cm,between origins}]
			\setwiretype{n}	& & & A & B & A & \\
			& \qw &\qw	& \ctrl{1} &\ctrl{1}  & \ctrl{1} & \qw \\
			& \gate{X} &\gate{X}	&\octrl{1} & \octrl{1} &\octrl{1}   & \qw \\
			& \qw &\qw	& \ctrl{1}&\targ{}  & \ctrl{1}  & \qw \\
			& \gate{X}&\gate{X}	&\targ{} &  \octrl{-1} &\targ{}  & \qw 
		\end{quantikz}}
		\overset{\text{Rule~\ref{rule-gateswap-derived},\ref{rule-swap-gate-changectrl}}}{\Longleftrightarrow}
	\]
	\[
		\scalebox{0.7}{
		\begin{quantikz}[row sep={0.7cm,between origins}]
			&\qw	& \ctrl{1} &\ctrl{1}  & \ctrl{1} &\qw & \qw \\
			&\gate{X}	&\ctrl{1} & \ctrl{1} &\ctrl{1}  & \gate{X} & \qw \\
			&\qw	& \ctrl{1}&\targ{}  & \ctrl{1}  &\qw& \qw \\
			&\gate{X}	&\targ{} &  \ctrl{-1} &\targ{} & \gate{X} & \qw 
		\end{quantikz}}
		\overset{\text{Rule~\ref{rule-ABA-BAB}}}{\Longleftrightarrow}
		\scalebox{0.7}{
		\begin{quantikz}[row sep={0.7cm,between origins}]
			&\qw	& \ctrl{1} &\ctrl{1}  & \ctrl{1} &\qw & \qw \\
			& \gate{X}	&\ctrl{1} & \ctrl{1} &\ctrl{1}  & \gate{X} & \qw \\
			&\qw	&\targ{}  & \ctrl{1} &\targ{}   &\qw& \qw \\
			& \gate{X}	&\ctrl{-1}& \targ{}  &\ctrl{-1} & \gate{X} & \qw 
		\end{quantikz}
	}
	\]
	\[
		\overset{\text{Rule~\ref{rule-gateswap-derived},\ref{rule-swap-gate-changectrl}}}{\Longleftrightarrow}
		\scalebox{0.7}{
		\begin{quantikz}[row sep={0.7cm,between origins}]
			\setwiretype{n}	& & & B & A & B & \\
			&\qw &\qw	& \ctrl{1} &\ctrl{1}  & \ctrl{1} & \qw \\
			& \gate{X} & \gate{X}	&\octrl{1} & \octrl{1} &\octrl{1}  & \qw \\
			&\qw &\qw	&\targ{}  & \ctrl{1} &\targ{}   & \qw \\
			& \gate{X} & \gate{X}	&\octrl{-1}& \targ{}  &\octrl{-1}  & \qw 
		\end{quantikz}}
		\overset{\text{Rule~\ref{rule-gate-elimn}}}{\Longleftrightarrow}
		\scalebox{0.7}{
		\begin{quantikz}[row sep={0.7cm,between origins}]
			\setwiretype{n}	& \push{B}& \push{A}& \push{B}  & \\
			&\ctrl{1} &\ctrl{1}  & \ctrl{1} & \qw \\
			&\octrl{1} & \octrl{1} &\octrl{1}   & \qw \\
			\lstick{$q$:} &\targ{}  & \ctrl{1} &\targ{}  & \qw \\
			\lstick{$p$:} &\octrl{-1}& \targ{}  &\octrl{-1} & \qw 
	\end{quantikz}}
	\]
\end{proof}

\begin{lemma}\label{lem-ABCDCBA}
	Let $A_1,\dots,A_{m}$ be $m$ $n$-bit MPMCT gates. If the following two conditions are satisfied
	\begin{enumerate}[(1)]
		\item\label{enu-ABCDCBA-1} for $1\leq i < m$, $A_i$ and $A_{i+1}$ coincide on the polarities of their common control bits,
		\item\label{enu-ABCDCBA-2} for $A_i$ and $A_j$ with $|i-j|\geq 2$, there is a control bit that has different polarities in $A_i$ and $A_j$, respectively,
	\end{enumerate}
	then
	\[
	\begin{aligned}
		& A_1 A_2\cdots A_{m-1}A_m A_{m-1}\cdots A_2A_1 \\
		\Leftrightarrow\; & A_m A_{m-1}\cdots A_2 A_1 A_2\cdots A_{m-1} A_m.
	\end{aligned}
	\]
\end{lemma}
\begin{proof}
	We show the two main steps for the transformation in an example below.
	
	\textbf{Step\hspace{0.1cm}1:} By \eqref{enu-ABCDCBA-1} and Lemma~\ref{lem-ABA-BAB}, we have $A_i A_{i+1} A_i \Leftrightarrow A_{i+1} A_i A_{i+1}$ ($1\leq i < m$). The circuit can be transformed from the inside as follows:
	\[
	\begin{aligned}
		& A_1 \cdots A_{m-2} A_{m-1} A_m A_{m-1} A_{m-2} \cdots A_1 \\
		\Leftrightarrow\; &  A_1 \cdots A_{m-2} A_m A_{m-1} A_m A_{m-2} \cdots A_1.
	\end{aligned}
	\]
	
	\textbf{Step\hspace{0.1cm}2:} By \eqref{enu-ABCDCBA-2} and Rule~\ref{rule-gateswap}, we have $A_i A_j \Leftrightarrow  A_j A_i$ ($|i-j|\geq 2$). Thus, the two $A_m$ gates can be moved to the outside of the circuit:
	\[
	\begin{aligned}
	& A_1 \cdots A_{m-2} A_m A_{m-1} A_m A_{m-2} \cdots A_1 \\
	\Leftrightarrow\; & A_m A_1 \cdots A_{m-2} A_{m-1} A_{m-2} \cdots A_1 A_m.
	\end{aligned}
	\]
	Therefore, by repeating Step\hspace{0.1cm}1 and Step\hspace{0.1cm}2, we can transform the two reversible circuits $A_1 A_2\cdots A_{m-1}A_m A_{m-1}\cdots A_2A_1$ and $A_m A_{m-1}\cdots A_2 A_1 A_2\cdots A_{m-1} A_m$ into each other.
\end{proof}

Let $\mathbf{C}$ be a reversible circuit. We denote by $\Delta(\mathbf{C})$ the set of gates that occur in $\mathbf{C}$.

\begin{lemma}\label{lem-MM-M}
	Let $M_i\in \Delta_\mathbb{H}$, and $\mathbf{D}$ a reversible circuit such that $\Delta(\mathbf{D})\subseteq \Delta_\mathbb{H}$ and $M_i\notin \Delta(\mathbf{D})$. Then the circuit $M_i\mathbf{D} M_i$ can be transformed into a circuit $\mathbf{D}'$ that has at most one occurrence of $M_i$ and $\Delta(\mathbf{D}')\subseteq \Delta(\mathbf{D})\cup \{M_i\}$.
\end{lemma}
\begin{proof}
	First, we consider two simple cases for $M_i\mathbf{D} M_i$.
	\begin{enumerate}[(i)]
		\item\label{enu-MM-1} If $|i-j|\geq 2$ for every $M_j\in \Delta(\mathbf{D})$, then we can move the two $M_i$ gates to be adjacent by Lemma~\ref{lem-MMcommute} and Rule~\ref{rule-gateswap}, and eliminate them by Rule~\ref{rule-gate-elimn} to obtain $\mathbf{D}'$.
		\item\label{enu-MM-2}  If $M_i\mathbf{D} M_i$ is in the form of
		\[
		M_i M_{i\circ 1}\cdots M_{i\circ(k-1)} M_{i\circ k} M_{i\circ (k-1)} \cdots M_{i\circ 1} M_i,
		\]
		where $\circ \in\{+,-\}$, then by Lemma~\ref{lem-MMcommute} and \ref{lem-ABCDCBA}, we can transform $M_i\mathbf{D} M_i$ into
		\[
		M_{i\circ k} M_{i\circ (k-1)} \cdots M_{i\circ 1} M_i M_{i\circ 1}\cdots M_{i\circ (k-1)} M_{i\circ k},
		\]
		which has exactly one occurrence of $M_i$.
	\end{enumerate}

The basic idea of the proof is to transform $M_i\mathbf{D} M_i$ into a circuit $\mathbf{D}_1 M_i\mathbf{D}_2 M_i \mathbf{D}_3$ such that $\mathbf{D}_1, \mathbf{D}_3$ do not have any occurrence of $M_i$, and $M_i\mathbf{D}_2 M_i$ satisfies the condition in \eqref{enu-MM-1} or \eqref{enu-MM-2}. Then the circuit $\mathbf{D}'$ can be obtained immediately. We list the four cases of the transformation and show how to deal with them in the following. For simplicity, we only consider the subcircuit between the two $M_i$ gates.
	 
\textbf{Case\hspace{0.1cm}1:} The circuit has the form $M_i \cdots M_j M_j\cdots M_i$. Then the two gates $M_j M_j$ can be removed by Rule~\ref{rule-gate-elimn}.

\textbf{Case\hspace{0.1cm}2:} The circuit has the form 
	\[
	M_i M_{i\circ 1}\cdots M_{i\circ(k-1)} M_{i\circ k} \underline{M_x} \cdots M_i
	\]
	or
	\[
	M_i\cdots \underline{M_x} M_{i\circ k} M_{i\circ(k-1)} \cdots M_{i\circ 1} M_i,
	\]
	where $\circ \in\{+,-\}$ and $|x- j|\geq 2$ for every $j$ among the numbers $i,i\circ 1,\dots, i\circ k$. By Lemma~\ref{lem-MMcommute} and Rule~\ref{rule-gateswap}, we can move $M_x$ to the outside of the circuit and obtain 
	\[
	\underline{M_x} M_i M_{i\circ 1}\cdots M_{i\circ(k-1)} M_{i\circ k} \cdots M_i
	\]
	or
	\[
	M_i\cdots M_{i\circ k} M_{i\circ(k-1)} \cdots M_{i\circ 1} M_i \underline{M_x}.
	\]

\textbf{Case\hspace{0.1cm}3:} The circuit has the form 
	\[
	M_i M_{i\circ 1}\cdots  M_{i\circ(k-2)} \underline{M_{i\circ(k-1)} M_{i\circ k}  M_{i\circ(k-1)}} \cdots M_i,
	\]
	where $\circ \in\{+,-\}$. By Lemma~\ref{lem-ABA-BAB} we have
	\[
	M_{i\circ(k-1)} M_{i\circ k}  M_{i\circ(k-1)} \Leftrightarrow M_{i\circ k}  M_{i\circ(k-1)} M_{i\circ k}.
	\]
	Hence, the circuit can be transformed into 
	\[
	M_i M_{i\circ 1}\cdots M_{i\circ(k-2)} \underline{M_{i\circ k}  M_{i\circ(k-1)} M_{i\circ k}} \cdots M_i,
	\]
	which satisfies the condition in Case\hspace{0.1cm}2. Thus, we can move the first $M_{i\circ k}$ to the left side of the circuit and obtain
	\[
	\underline{M_{i\circ k}} M_i M_{i\circ 1}\cdots M_{i\circ(k-2)} \underline{M_{i\circ(k-1)} M_{i\circ k}} \cdots M_i.
	\]

\textbf{Case\hspace{0.1cm}4:} The circuit has the form
	\[
	M_i M_{i\circ 1}\cdots \underline{M_{i\circ j} M_{i\circ (j+1)}} M_{i\circ (j+2)} \cdots M_{i\circ k}  \underline{M_{i\circ j}} \cdots M_i,
	\]
	where $\circ \in\{+,-\}$ and $k-j \geq 2$. By Lemma~\ref{lem-MMcommute} and Rule~\ref{rule-gateswap}, we can move the second $M_{i\circ j}$ gate to the right side of $M_{i\circ (j+1)}$ and obtain
	\[
	M_i M_{i\circ 1}\cdots\underline{M_{i\circ j} M_{i\circ (j+1)} M_{i\circ j}}  M_{i\circ (j+2)} \cdots M_{i\circ k} \cdots M_i,
	\]
	which satisfies the condition in Case\hspace{0.1cm}3. Thus, the circuit above can be transformed into the following circuit
	\[
	M_{i\circ (j+1)} M_i M_{i\circ 1}\cdots M_{i\circ j} M_{i\circ (j+1)}\cdots M_{i\circ k} \cdots M_i.
	\]

We do the transformation according to the four cases above. Note that in each case, either two gates are eliminated or one gate is moved to the outside of the two $M_i$ gates. Hence, we can eventually get a subcircuit $M_i\mathbf{D}_2 M_i$ that satisfies the condition in \eqref{enu-MM-1} or \eqref{enu-MM-2}. Moreover, since no new gate is introduced in the transformation, we have $\Delta(\mathbf{D}')\subseteq \Delta(\mathbf{D})\cup \{M_i\}$.
\end{proof}

\begin{lemma}\label{lem-M123}
	Let $M_i\in \Delta_\mathbb{H}$, and $\mathbf{D}$ a reversible circuit such that $\Delta(\mathbf{D})\subseteq \Delta_\mathbb{H}$ and $j>i$ for every $M_j\in \Delta(\mathbf{D})$. Then the circuit $M_i\mathbf{D}$ can be transformed into a circuit 
	\[
	\mathbf{D}'' = \mathbf{D}' M_i M_{i+1}\cdots M_{i+k},
	\]
	where the subcircuit $\mathbf{D}'$ does not have any occurrence of $M_i$ and $\Delta(\mathbf{D}'') \subseteq \Delta(\mathbf{D})\cup \{M_i\}$.
\end{lemma}
\begin{proof}
	By an analysis of the circuit $\mathbf{D}''$ we see that the Case\hspace{0.1cm}1, 2, 3, and 4 in the proof of Lemma~\ref{lem-MM-M} can be adopted for the transformation from $M_i\mathbf{D}$ to $\mathbf{D}''$, and the lemma follows.
\end{proof}

\begin{proposition}\label{prop-canonical-lim}
	 Every reversible circuit $\mathbf{C}$ with $\Delta(\mathbf{C})\subseteq \Delta_\mathbb{H}$ can be transformed into its canonical form based on $\mathbb{H}$.
\end{proposition}
\begin{proof}
	 By~\eqref{fact-gatenumber} of Fact~\ref{fact-canonicalform}, we know that the gate $M_0$ occurs at most once in the canonical form of $\mathbf{C}$. 
	 By Lemma~\ref{lem-MM-M}, we can transform $\mathbf{C}$ into a circuit $\mathbf{C}'$ that has at most one occurrence of $M_0$. 
	 \begin{enumerate}[(i)]
	 \item If $\mathbf{C}'$ has one occurrence of $M_0$, then by Lemma~\ref{lem-M123}, $\mathbf{C}'$ can be transformed into a circuit of the form $\mathbf{C}''\mathbf{C}_0$, where $\mathbf{C}_0=M_0 M_1 \cdots M_i$ ($i\geq 0$) and for every gate $M_x$ in $\Delta(\mathbf{C}'')$, $x>0$.
	 \item\label{enu-canonical-search} If $\mathbf{C}'$ has no occurrence of $M_0$, then we continue the transformation and reduce the gates $M_1,M_2\dots,M_{2^n-2}$ in turn according to Lemma~\ref{lem-MM-M} until finding a gate $M_j$ that has exactly one occurrence. By Lemma~\ref{lem-M123}, $\mathbf{C}'$ can be transformed into a circuit $\mathbf{C}''\mathbf{C}_0$, where $\mathbf{C}_0=M_j M_{j+1} \cdots M_{j+k}$ ($k\geq 0$) and for every gate $M_x$ in $\Delta(\mathbf{C}'')$, $x>j$.
	 \end{enumerate}
	 
	Suppose that the circuit $\mathbf{D} \mathbf{C}_i\cdots\mathbf{C}_1 \mathbf{C}_0$ has been constructed, and $M_x$ is the first gate of $\mathbf{C}_i$. By~\eqref{fact-firstgate} of Fact~\ref{fact-canonicalform}, similarly as in~\eqref{enu-canonical-search}, we continue the transformation on the subcircuit $\mathbf{D}$ by reducing the gates $M_{x+1},M_{x+2},\dots,M_{2^n-2}$ in turn until a gate with only one occurrence is found. And then construct the subcircuit $\mathbf{C}_{i+1}$ by Lemma~\ref{lem-M123}. Finally, we can get the canonical form $\mathbf{C}_m\cdots\mathbf{C}_1 \mathbf{C}_0$ of $\mathbf{C}$.
\end{proof}

\begin{example}
We show how to transform a circuit $\mathbf{C}$ with 3 bits into the canonical form. We use the Hamiltonian path $\mathbb{H}$ in Example~\ref{exam-Hamiltonian}. All gates of $\mathbf{C}$ are from $\Delta_\mathbb{H}$ (see Fig.~\ref{fig-hamiltonian}).
	\[\mathbf{C}=
	\scalebox{0.75}{\begin{quantikz}[row sep={0.7cm,between origins}, column sep=0.45cm]
	\setwiretype{n}	
	& \push{M_4}& \push{M_2} & \push{M_3} & \textcolor{red}{\push{M_0}} & \push{M_4} & \push{M_1}	& \push{M_3}  & \push{M_4} & \push{M_0} & \\
	& \ctrl{1}  & \octrl{1}  & \targ{} 	  & \octrl{1}   & \ctrl{1} 	 & \octrl{1}	& \targ{}	  & \ctrl{1}   & \octrl{1}	& \qw \\
	& \ctrl{1}  & \ctrl{1}   & \ctrl{-1}  & \octrl{1}	& \ctrl{1}   & \targ{}		& \ctrl{-1}   & \ctrl{1}   & \octrl{1}	& \qw \\
	& \targ{}   & \targ{}    & \octrl{-1} & \targ{} 	& \targ{}    & \ctrl{-1}	& \octrl{-1}  & \targ{}	   & \targ{}    & \qw
	\end{quantikz}}
	\]
	\[
	\overset{\text{Rule~\ref{rule-gateswap}}}{\Longleftrightarrow}
	\scalebox{0.75}{\begin{quantikz}[row sep={0.7cm,between origins}, column sep=0.45cm]
	\setwiretype{n}	
			& \push{M_4}& \push{M_2} & \push{M_3} & \push{M_4}	& \textcolor{red}{\push{M_0}} & \push{M_1}	& \textcolor{blue}{\push{M_3}}  & \textcolor{blue}{\push{M_4}} & \push{M_0} & \\
			& \ctrl{1}  & \octrl{1}  & \targ{} 	  & \ctrl{1}	& \octrl{1}  & \octrl{1}	& \targ{}	  & \ctrl{1}   & \octrl{1}	& \qw \\
			& \ctrl{1}  & \ctrl{1}   & \ctrl{-1}  & \ctrl{1}	& \octrl{1}  & \targ{}		& \ctrl{-1}   & \ctrl{1}   & \octrl{1}	& \qw \\
			& \targ{}   & \targ{}    & \octrl{-1} & \targ{} 	& \targ{}    & \ctrl{-1}	& \octrl{-1}  & \targ{}	   & \targ{}    & \qw
	\end{quantikz}}
	\]
	\[
	\overset{\text{Rule~\ref{rule-gateswap}}}{\Longleftrightarrow}
	\scalebox{0.75}{\begin{quantikz}[row sep={0.7cm,between origins}, column sep=0.45cm]
	\setwiretype{n}	
		& \push{M_4}& \push{M_2} & \push{M_3} & \push{M_4}	& \textcolor{blue}{\push{M_3}} & \textcolor{blue}{\push{M_4}} 	& \textcolor{red}{\push{M_0}} & \push{M_1} & \push{M_0} & \\
		& \ctrl{1}  & \octrl{1}  & \targ{} 	  & \ctrl{1}	& \targ{}	  & \ctrl{1} & \octrl{1}  & \octrl{1}	 & \octrl{1}	& \qw \\
		& \ctrl{1}  & \ctrl{1}   & \ctrl{-1}  & \ctrl{1}	& \ctrl{-1}   & \ctrl{1} & \octrl{1}  & \targ{}		 & \octrl{1}	& \qw \\
		& \targ{}   & \targ{}    & \octrl{-1} & \targ{} 	& \octrl{-1}  & \targ{}	 & \targ{}    & \ctrl{-1}	 & \targ{}  & \qw
	\end{quantikz}}
	\]
	\[
	\overset{\text{Lemma~\ref{lem-ABA-BAB}}}{\Longleftrightarrow}
	\scalebox{0.75}{\begin{quantikz}[row sep={0.7cm,between origins}, column sep=0.45cm]
	\setwiretype{n}	
	& \push{M_4}& \push{M_2} & \textcolor{blue}{\push{M_3}} & \textcolor{blue}{\push{M_4}}	& \textcolor{blue}{\push{M_3}} & \push{M_4}& \push{M_1} & \push{M_0} & \push{M_1} & \\
	& \ctrl{1}  & \octrl{1}  & \targ{} 	  & \ctrl{1}	& \targ{}	  & \ctrl{1}\slice{$\qquad\mathbf{C}_1$} & \octrl{1}\slice{$\qquad\mathbf{C}_0$}  & \octrl{1}  & \octrl{1}	& \qw \\
	& \ctrl{1}  & \ctrl{1}   & \ctrl{-1}  & \ctrl{1}	& \ctrl{-1}   & \ctrl{1} & \targ{}	  & \octrl{1}  & \targ{}	& \qw \\
	& \targ{}   & \targ{}    & \octrl{-1} & \targ{} 	& \octrl{-1}  & \targ{}	 & \ctrl{-1}  & \targ{}    & \ctrl{-1}	& \qw
	\end{quantikz}}
	\]
	\[
	\overset{\text{Lemma~\ref{lem-ABA-BAB}}}{\Longleftrightarrow}
	\scalebox{0.75}{\begin{quantikz}[row sep={0.7cm,between origins}, column sep=0.45cm]
	\setwiretype{n}	
	& \push{M_4}& \push{M_2} & \textcolor{blue}{\push{M_4}} & \textcolor{blue}{\push{M_3}}	& \textcolor{blue}{\push{M_4}} & \push{M_4}& \push{M_1} & \push{M_0} & \push{M_1} & \\
	& \ctrl{1}  & \octrl{1}  & \ctrl{1}	& \targ{}	 & \ctrl{1}  & \ctrl{1}\slice{$\qquad\mathbf{C}_1$} & \octrl{1}\slice{$\qquad\mathbf{C}_0$}  & \octrl{1}  & \octrl{1}	& \qw \\
	& \ctrl{1}  & \ctrl{1}   & \ctrl{1}	& \ctrl{-1}  & \ctrl{1}  & \ctrl{1} & \targ{}	  & \octrl{1}  & \targ{}	& \qw \\
	& \targ{}   & \targ{}    & \targ{} 	& \octrl{-1} & \targ{}  & \targ{}	 & \ctrl{-1}  & \targ{}    & \ctrl{-1}	& \qw
	\end{quantikz}}
	\]
	\[
	\overset{\text{Rule~\ref{rule-gate-elimn}}}{\Longleftrightarrow}
	\scalebox{0.75}{\begin{quantikz}[row sep={0.7cm,between origins}]
	\setwiretype{n}	
	& \push{M_4}& \push{M_2} & \push{M_4} & \push{M_3}	& \push{M_1} & \push{M_0} & \push{M_1} & \\
	& \ctrl{1}  & \octrl{1}  & \ctrl{1}	& \targ{}\slice{$\qquad\mathbf{C}_1$} & \octrl{1}\slice{$\qquad\mathbf{C}_0$}  & \octrl{1}  & \octrl{1}	& \qw \\
	& \ctrl{1}  & \ctrl{1}   & \ctrl{1}	& \ctrl{-1}  & \targ{}	  & \octrl{1}  & \targ{}	& \qw \\
	& \targ{}   & \targ{}    & \targ{} 	& \octrl{-1} & \ctrl{-1}  & \targ{}    & \ctrl{-1}	& \qw
	\end{quantikz}}
	\]
	\[
	\overset{\text{Rule~\ref{rule-gateswap}}}{\Longleftrightarrow}
	\scalebox{0.75}{\begin{quantikz}[row sep={0.7cm,between origins}]
	\setwiretype{n}	
	& \push{M_2}& \push{M_4} & \push{M_4} & \push{M_3}	& \push{M_1} & \push{M_0} & \push{M_1} & \\
	& \octrl{1}  & \ctrl{1}  & \ctrl{1}	& \targ{}\slice{$\qquad\mathbf{C}_1$} & \octrl{1}\slice{$\qquad\mathbf{C}_0$}  & \octrl{1}  & \octrl{1}	& \qw \\
	& \ctrl{1}  & \ctrl{1}   & \ctrl{1}	& \ctrl{-1}  & \targ{}	  & \octrl{1}  & \targ{}	& \qw \\
	& \targ{}   & \targ{}    & \targ{} 	& \octrl{-1} & \ctrl{-1}  & \targ{}    & \ctrl{-1}	& \qw
	\end{quantikz}}
	\]
	\[
	\overset{\text{Rule~\ref{rule-gate-elimn}}}{\Longleftrightarrow}
	\scalebox{0.75}{\begin{quantikz}[row sep={0.7cm,between origins}]
	\setwiretype{n}	
	& \push{M_2}& \push{M_3}	& \push{M_1} & \push{M_0} & \push{M_1} & \\
	\slice{$\qquad\mathbf{C}_2$}& \octrl{1} & \targ{}\slice{$\qquad\mathbf{C}_1$} & \octrl{1}\slice{$\qquad\mathbf{C}_0$}  & \octrl{1}  & \octrl{1}	& \qw \\
	& \ctrl{1}  & \ctrl{-1}  & \targ{}	  & \octrl{1}  & \targ{}	& \qw \\
	& \targ{}   & \octrl{-1} & \ctrl{-1}  & \targ{}    & \ctrl{-1}	& \qw
	\end{quantikz}}
	\]
The canonical form of $\mathbf{C}$ based on $\mathbb{H}$ is $\mathbf{C}_2 \mathbf{C}_1 \mathbf{C}_0$, where $\mathbf{C}_0=M_0 M_1$, $\mathbf{C}_1=M_1$, and $\mathbf{C}_2=M_2 M_3$.
\end{example}

\subsection{Completeness of the rules}
In this section, we prove a generalization of Proposition~\ref{prop-canonical-lim} that every reversible circuit can be transformed into its unique canonical form, which implies that $\mathcal{RC}$ is complete.

A coordinate sequence is a sequence of numbers from $\{1,\dots,n\}$. Let $\omega =(m_1,m_2,\dots,m_k)$ be a coordinate sequence, and $b_0\in \{0,1\}^n$. We say that $\omega$ generates a string $b_k$ from $b_0$ if there is a sequence $(b_0,b_1,\dots,b_k)$ of strings such that $b_{i+1}$ is obtained by flipping the $m_{(i+1)}$-th bit of $b_i$ ($0\leq i < k$), as shown below
\[
\omega\!:\ b_0 \overset{m_1}{\longrightarrow} b_1 \overset{m_2}{\longrightarrow} b_2 \overset{m_3}{\longrightarrow} \cdots \overset{m_k}{\longrightarrow} b_k.
\]

\begin{example}\label{exam-coordinate}
	Let $b_0=000$ and $\omega=(1,2,1,2,3)$. We have
	\[
	\omega\!:\  000 \overset{1}{\longrightarrow} 100 \overset{2}{\longrightarrow} 110 \overset{1}{\longrightarrow} 010 \overset{2}{\longrightarrow} 000 \overset{3}{\longrightarrow} 001.
	\]
	Let $\omega_1=(1,2,2,1,3)$ that swaps the 3rd and 4th elements in $\omega$. We have 
	\[
	\omega_1\!:\  000 \overset{1}{\longrightarrow} 100 \overset{2}{\longrightarrow} 110 \overset{2}{\longrightarrow} 100 \overset{1}{\longrightarrow} 000 \overset{3}{\longrightarrow} 001.
	\]
	Let $\omega_2=(1,1,3)$ that deletes the number 2 in $\omega_1$. We have 
	\[
	\omega_2\!:\  000 \overset{1}{\longrightarrow} 100 \overset{1}{\longrightarrow} 000 \overset{3}{\longrightarrow} 001.
	\]
	Let $\omega_3=(3)$ that deletes the number 1 in $\omega_2$. We have 
	\[
	\omega_3\!:\ 000 \overset{3}{\longrightarrow} 001.
	\]
	It is easy to check that the coordinate sequences $\omega,\omega_1,\omega_2,\omega_3$ generate the same string from $b_0$.
\end{example}

\begin{fact}\label{fact-coordinate}
	Let $\omega$ be a coordinate sequence.
	\begin{enumerate}[(1)]
		\item Changing the order of elements in $\omega$ does not change the generated string.
		\item Deleting two adjacent identical elements in $\omega$ does not change the generated string.
	\end{enumerate}
\end{fact}

By Fact~\ref{fact-coordinate}, we can reduce $\omega$ to a coordinate sequence $\omega'$ such that every number in $\omega$ has at most one occurrence in $\omega'$, and $\omega'$ generates the same string as $\omega$ generates. More precisely, all elements that have an even number of occurrences in $\omega$ can be deleted, and the other elements that have an odd number of occurrences only keep one occurrence. The proof is straightforward, since if a number $m$ occurs $h$ times in $\omega$, then the $m$-th bit of $b_0$ will flip $h$ times.

\begin{theorem}\label{thm-circuitcanonical}
	Every $n$-bit reversible circuit can be transformed into its canonical form based on $\mathbb{H}$.
\end{theorem}
\begin{proof}
	By Rule~\ref{rule-qubit-add-remove}, every $m$-bit MPMCT gate ($m<n$) can be transformed into a circuit that consists of $n$-bit MPMCT gates, i.e., every $n$-ary reversible circuit can be transformed into a reversible circuit that only contains $n$-bit MPMCT gates. By Proposition~\ref{prop-canonical-lim}, it will thus be sufficient to prove that every $n$-bit MPMCT gate that is not in $\Delta_\mathbb{H}$ can be transformed into a reversible circuit that only consists of the gates in $\Delta_\mathbb{H}$.
	
	Suppose that $M$ is an $n$-bit MPMCT gate that is not in $\Delta_\mathbb{H}$, and it exchanges the two strings $a_i,a_j$ in $\mathbb{H}$ ($0\leq i< j \leq 2^n-1$). Therefore, $M$ is equivalent to the circuit
	\[
	\mathbf{C}=M_{i} M_{i+1}\cdots M_{j-2} M_{j-1} M_{j-2}\cdots M_{i+1} M_{i}.
	\]
	We show how to transform the circuit $\mathbf{C}$ into $M$. The sequence $(a_i,a_{i+1},\cdots,a_{j-1},a_j)$ defines a coordinate sequence
	\[
	\omega=(m_i,m_{i+1},\cdots,m_{j-1}),
	\]
	where for each $i\leq k < j$, $a_k,a_{k+1}$ only differ in the $m_k$-th bit. It is obvious that the target bit of $M_k$ is $q_{m_k}$ ($i\leq k < j$).
	The (left half part of) circuit $\mathbf{C}$ corresponds to the generating process of $\omega$ from $a_i$.
	
	If $a_i,a_j$ differ in the $m$-th bit, then $m$ must occur an odd times in $\omega$, and all other numbers in $\omega$ occur an even times. Hence, $\omega$ can be reduced to the sequence $(m)$ by Fact~\ref{fact-coordinate} (see Example~\ref{exam-coordinate}). To make the proof more understandable, we show the transformation for the circuit $\mathbf{C}$ according to the moving and deleting actions on $\omega$. 
	
	Let $m_g = m_h$ ($g < h$) be two elements in $\omega$ such that all numbers $m_{g+1}, \dots, m_{h-1}$ between $m_g$ and $m_h$ are different from each other, and none of them equals $m_g$. Thus, $\omega$ can be written as
	\[
	(m_i\cdots m_{g-1}, m_{g}, m_{g+1}, \dots, m_{h-1}, m_{h}, m_{h+1}, \dots, m_{j-1}).
	\]
	We move $m_h$ to the right side of $m_g$, and obtain the sequence
	\[
	(m_i\cdots m_{g-1}, m_{g}, m_{h}, m_{g+1}, \dots, m_{h-1}, m_{h+1}, \dots, m_{j-1}).
	\] 
	Then we delete $m_{g}, m_{h}$ from the sequence to get a new coordinate sequence
	\[
	\omega_1 = (m_i\cdots m_{g-1}, m_{g+1}, \dots, m_{h-1}, m_{h+1}, \dots, m_{j-1}),
	\]
	which also generates $a_j$ from $a_i$ by Fact~\ref{fact-coordinate}.
	
	Let $M_g,M_h$ be the corresponding gates for $m_g,m_h$, respectively. In the following, we show the transformation on $\mathbf{C}$. 
	Let $M'$ be a gate, $\mathbf{D}$ a circuit and $\mathbf{D}^{-1}$ its inverse. For simplicity of notation, we denote the circuit $\mathbf{D} M'\mathbf{D}^{-1}$ by $\mathbf{D} M'\|$, e.g.,
	\[
	\mathbf{C} = M_i\cdots M_{g-1} M_{g} M_{g+1} \cdots M_{h-1} M_{h} M_{h+1} \cdots M_{j-1}\|.
	\]
 	By Lemma~\ref{lem-ABCDCBA}, we have
 	\[
 	\begin{aligned}
 		& M_{h} M_{h+1} \cdots M_{j-1} \cdots M_{h+1} M_{h} \\
 	\Leftrightarrow  & M_{j-1} \cdots M_{h+1} M_{h} M_{h+1} \cdots M_{j-1}.
 	\end{aligned}
 	\]
	The circuit $\mathbf{C}$ can be transformed into
	\[
	M_i\cdots M_{g-1} \underline{M_{g} M_{g+1} \cdots M_{h-1}} M_{j-1} \cdots M_{h+1} M_{h} \| .
	\]
	By Lemma~\ref{lem-MMcommute}, for any $M_{k_1}\in \{M_{g},M_{g+1}, \dots, M_{h-1}\}$ and $M_{k_2}\in \{M_{h+1},\dots, M_{j-1}\}$, there is a control bit that has different polarities in $M_{k_1}$ and $M_{k_2}$, respectively. Hence, by Rule~\ref{rule-gateswap} we can move $M_{g} M_{g+1} \cdots M_{h-1}$ to the left side of $M_h$
	\begin{equation}\label{eqn-coord1}
	M_i\cdots M_{g-1} M_{j-1} \cdots M_{h+1} \underline{M_{g} M_{g+1} \cdots M_{h-1}} M_{h} \|.
	\end{equation}
 	By Lemma~\ref{lem-ABCDCBA}, we have
 	\[
 	\begin{aligned}
 		& M_{g+1} \cdots M_{h-1} M_{h} M_{h-1} \cdots M_{g+1} \\
 		\Leftrightarrow  & M_{h} M_{h-1} \cdots M_{g+1} \cdots M_{h-1} M_{h}.
 	\end{aligned}
 	\]
 	The circuit \eqref{eqn-coord1} can be transformed into
 	\[
 	M_i\cdots M_{g-1} M_{j-1} \cdots M_{h+1} M_{g} M_{h} M_{h-1} \cdots M_{g+1}\|.
 	\]
 	By Rule~\ref{rule-swap-gate-changectrl} and Rule~\ref{rule-completeness}, $M_{g} M_{h}$ can be removed from the circuit. We obtain
 	\begin{equation}\label{eqn-coord2}
 		M_i\cdots M_{g-1} M_{j-1} \cdots M_{h+1} M_{h-1}' \cdots M_{g+1}'\|
 	\end{equation}
 	where $M_{h-1}', \dots, M_{g+1}'$ are obtained by changing the polarity of bit $q_{m_g}$ in $M_{h-1},\dots, M_{g+1}$, respectively.
 	Here we use Rule~\ref{rule-swap-gate-changectrl} to transform the circuit so that the condition in Rule~\ref{rule-completeness} can be satisfied. More precisely, we use an X gate to pass through a line in the circuit such that the polarity of the control bit in the line is changed. Note that the selection of $m_g$ and $m_h$ also ensures that the condition in Rule~\ref{rule-completeness} is met.
 	
 	We now apply Lemma~\ref{lem-ABCDCBA} again on the circuit \eqref{eqn-coord2} to get
 	\begin{equation}\label{eqn-coord3}
 		M_i\cdots M_{g-1} M_{g+1}' \cdots M_{h-1}' M_{h+1} \cdots M_{j-1} \|,
 	\end{equation}
 	which corresponds to the generating process of $\omega_1$ from $a_i$.
	If $m_{g-1} = m_{g+1}$ (resp. $m_{h-1} = m_{h+1}$), then $M_{g-1} = M_{g+1}'$ (resp. $M_{h-1} = M_{h+1}'$). We delete $m_{g-1}, m_{g+1}$ (resp. $m_{h-1}, m_{h+1}$) from $\omega_1$, and delete $M_{g-1},M_{g+1}'$ (resp. $M_{h-1}, M_{h+1}'$) from the circuit \eqref{eqn-coord3}. We check the coordinate sequence and delete the elements and gates until no adjacent identical element exists in the coordinate sequence.
	
	The new coordinate sequence and circuit can be dealt with as that for $\omega$ and $\mathbf{C}$. We repeat the procedure until the gate $M$ is obtained.
\end{proof}

\begin{example}
	Let $\mathbb{H}$ be the Hamiltonian path and $\Delta_\mathbb{H}$ the set of gates given in Example~\ref{exam-Hamiltonian}. We show how to transform the circuit in Example~\ref{exa-template} into the canonical circuit based on $\mathbb{H}$. We first transform it into a circuit whose gates are all from $\Delta_\mathbb{H}$, then transform it into the canonical form.
	\[
	\scalebox{0.75}{\begin{quantikz}[row sep={0.7cm,between origins}, column sep=0.4cm]
		\setwiretype{n}	&  &   & \\
			& \ctrl{1} & \qw	  & \qw \\
			& \targ{}  & \ctrl{1} & \qw \\
			& \qw	   & \targ{}  & \qw
	\end{quantikz}}
	\overset{\text{Rule~\ref{rule-qubit-add-remove}}}{\Longleftrightarrow} 
	\scalebox{0.75}{\begin{quantikz}[row sep={0.7cm,between origins}, column sep=0.4cm]
			\setwiretype{n}	
			& 			 &\push{M_5}& \push{M_2}& \push{M_4}& \\
			& \ctrl{1} 	 & \ctrl{1} & \octrl{1}	& \ctrl{1}	& \qw \\
			& \targ{}  	 & \targ{}	& \ctrl{1}  & \ctrl{1}  & \qw \\
			& \octrl{-1} & \ctrl{-1}& \targ{}   & \targ{}   & \qw
	\end{quantikz}}
	\]
	\[
	\overset{\text{Rule~\ref{rule-MPMCTtoMCT}}}{\Longleftrightarrow} 
	\scalebox{0.75}{\begin{quantikz}[row sep={0.7cm,between origins}, column sep=0.4cm]
	\setwiretype{n}	
	& 			&\push{M_5}  &	 		&\push{M_5}	& \push{M_2}& \push{M_4}& \\
	&\qw		& \ctrl{1} 	 &\qw		&\ctrl{1} 	& \octrl{1}	& \ctrl{1}	& \qw \\
	&\qw		& \targ{}  	 &\qw		&\targ{}	& \ctrl{1}  & \ctrl{1}  & \qw \\
	&\gate{x}	& \ctrl{-1}  &\gate{x}  &\ctrl{-1} 	& \targ{}   & \targ{}   & \qw
	\end{quantikz}}
	\overset{\text{Rule~\ref{rule-binary-generate}}}{\Longleftrightarrow} 
	\]
	\[		
	\scalebox{0.75}{\begin{quantikz}[row sep={0.7cm,between origins}, column sep=0.4cm]
		\setwiretype{n}	
		&\textcolor{red}{\push{M_0}} & \push{M_6} &\textcolor{red}{\push{M_2}} &\push{M_4}  & \push{M_5} 	&\textcolor{red}{\push{M_0}} & \push{M_6} &\textcolor{red}{\push{M_2}} &\push{M_4}&\push{M_5}	& \textcolor{red}{\push{M_2}}& \push{M_4}& \\
		&\octrl{1}  & \ctrl{1} 	 &\octrl{1}  &\ctrl{1} 	  & \ctrl{1} 	&\octrl{1}  & \ctrl{1} 	 &\octrl{1}  &\ctrl{1} 	&\ctrl{1} 	& \octrl{1}	& \ctrl{1}	& \qw \\
		&\octrl{1}  & \octrl{1}  &\ctrl{1} 	 &\ctrl{1}	  & \targ{}  	&\octrl{1}  & \octrl{1}  &\ctrl{1} 	 &\ctrl{1} 	&\targ{}	&\ctrl{1}   & \ctrl{1}  & \qw \\
		&\targ{}    & \targ{} 	 &\targ{}    &\targ{}	  & \ctrl{-1} 	&\targ{}    & \targ{} 	 &\targ{}    &\targ{}	&\ctrl{-1} 	& \targ{}   & \targ{}   & \qw
	\end{quantikz}}
	\]
	\[
	\overset{\text{Rule~\ref{rule-gateswap}}}{\Longleftrightarrow}
	\]
	\[
	\scalebox{0.75}{\begin{quantikz}[row sep={0.7cm,between origins}, column sep=0.4cm]
			\setwiretype{n}	
			& \push{M_6}  &\push{M_4} & \push{M_5} & \push{M_6}  &\push{M_4}&\push{M_5}	& \push{M_4} &\textcolor{red}{\push{M_0}} &\textcolor{red}{\push{M_0}} & \textcolor{red}{\push{M_2}} & \textcolor{red}{\push{M_2}} & \textcolor{red}{\push{M_2}}& \\
			& \ctrl{1} 	  &\ctrl{1}	  & \ctrl{1}   & \ctrl{1} 	 &\ctrl{1} 	&\ctrl{1} 	& \ctrl{1}	 &\octrl{1} &\octrl{1}  & \octrl{1} & \octrl{1} & \octrl{1}	& \qw \\
			& \octrl{1}   &\ctrl{1}	  & \targ{}	   & \octrl{1}   &\ctrl{1} 	&\targ{}	& \ctrl{1}   &\octrl{1} &\octrl{1}	& \ctrl{1}  & \ctrl{1}	& \ctrl{1}  & \qw \\
			& \targ{} 	  &\targ{}	  & \ctrl{-1}  & \targ{} 	 &\targ{}	&\ctrl{-1} 	& \targ{}    &\targ{} 	&\targ{} 	& \targ{}   & \targ{} 	& \targ{}   & \qw
	\end{quantikz}}
	\]
	\[
	\overset{\text{Rule~\ref{rule-gate-elimn}}}{\Longleftrightarrow} 
	\scalebox{0.75}{\begin{quantikz}[row sep={0.7cm,between origins}, column sep=0.4cm]
			\setwiretype{n}	
			& \push{M_6}  &\push{M_4} & \push{M_5} & \push{M_6}  &\textcolor{blue}{\push{M_4}} &\textcolor{blue}{\push{M_5}}&\textcolor{blue}{\push{M_4}} & \push{M_2}& \\
			& \ctrl{1} 	  &\ctrl{1}	  & \ctrl{1}   & \ctrl{1} 	 &\ctrl{1} 	&\ctrl{1} 	& \ctrl{1}	 & \octrl{1} & \qw \\
			& \octrl{1}   &\ctrl{1}	  & \targ{}	   & \octrl{1}   &\ctrl{1} 	&\targ{}	& \ctrl{1} 	 & \ctrl{1}  & \qw \\
			& \targ{} 	  &\targ{}	  & \ctrl{-1}  & \targ{} 	 &\targ{}	&\ctrl{-1} 	& \targ{}  	 & \targ{}   & \qw
	\end{quantikz}}
	\]
	\[
	\overset{\text{Lemma~\ref{lem-ABA-BAB}}}{\Longleftrightarrow}
	\scalebox{0.75}{\begin{quantikz}[row sep={0.7cm,between origins}, column sep=0.4cm]
			\setwiretype{n}	
			& \push{M_6}  &\push{M_4} & \textcolor{red}{\push{M_5}} & \textcolor{red}{\push{M_6}}  &\textcolor{red}{\push{M_5}} & \push{M_4} & \push{M_5}& \push{M_2}& \\
			& \ctrl{1} 	  &\ctrl{1}	  & \ctrl{1}   & \ctrl{1} 	 &\ctrl{1} 	& \ctrl{1} 	 & \ctrl{1}	 & \octrl{1} & \qw \\
			& \octrl{1}   &\ctrl{1}	  & \targ{}	   & \octrl{1}   &\targ{}	& \ctrl{1}   & \targ{}	 & \ctrl{1}  & \qw \\
			& \targ{} 	  &\targ{}	  & \ctrl{-1}  & \targ{} 	 &\ctrl{-1} & \targ{}    & \ctrl{-1} & \targ{}   & \qw
	\end{quantikz}}
	\]
	\[
	\overset{\text{Lemma~\ref{lem-ABA-BAB}}}{\Longleftrightarrow}
	\scalebox{0.75}{\begin{quantikz}[row sep={0.7cm,between origins}, column sep=0.4cm]
			\setwiretype{n}	
			& \textcolor{blue}{\push{M_6}}  &\textcolor{blue}{\push{M_4}} & \push{M_6}  &\push{M_5} & \push{M_6} 	& \push{M_4} & \push{M_5}& \push{M_2}& \\
			& \ctrl{1} 	  &\ctrl{1}	  & \ctrl{1} 	&\ctrl{1} 	& \ctrl{1} 		& \ctrl{1} 	 & \ctrl{1}	 & \octrl{1} & \qw \\
			& \octrl{1}   &\ctrl{1}	  & \octrl{1}   &\targ{}	& \octrl{1} 	& \ctrl{1}   & \targ{}	 & \ctrl{1}  & \qw \\
			& \targ{} 	  &\targ{}	  & \targ{} 	&\ctrl{-1} 	& \targ{} 		& \targ{}    & \ctrl{-1} & \targ{}   & \qw
	\end{quantikz}}
	\]
	\[
	\overset{\text{Rule~\ref{rule-gateswap}}}{\Longleftrightarrow}
	\scalebox{0.75}{\begin{quantikz}[row sep={0.7cm,between origins}, column sep=0.4cm]
			\setwiretype{n}	
			& \push{M_4}  &\textcolor{red}{\push{M_6}} & \textcolor{red}{\push{M_6}}  &\push{M_5} & \push{M_6} 	& \push{M_4} & \push{M_5}& \push{M_2}& \\
			& \ctrl{1} 	  &\ctrl{1}	  & \ctrl{1} 	&\ctrl{1} 	& \ctrl{1} 		& \ctrl{1} 	 & \ctrl{1}	 & \octrl{1} & \qw \\
			& \ctrl{1}    &\octrl{1}  & \octrl{1}   &\targ{}	& \octrl{1} 	& \ctrl{1}   & \targ{}	 & \ctrl{1}  & \qw \\
			& \targ{} 	  &\targ{}	  & \targ{} 	&\ctrl{-1} 	& \targ{} 		& \targ{}    & \ctrl{-1} & \targ{}   & \qw
	\end{quantikz}}
	\]
	\[
	\overset{\text{Rule~\ref{rule-gate-elimn}}}{\Longleftrightarrow} 
	\scalebox{0.75}{\begin{quantikz}[row sep={0.7cm,between origins}, column sep=0.4cm]
			\setwiretype{n}	
			& \push{M_4}  &\push{M_5} & \textcolor{blue}{\push{M_6}} & \textcolor{blue}{\push{M_4}} & \push{M_5}& \push{M_2}& \\
			& \ctrl{1} 	  &\ctrl{1}   & \ctrl{1} 	& \ctrl{1} 	 & \ctrl{1}	 & \octrl{1} & \qw \\
			& \ctrl{1}    &\targ{}	  & \octrl{1} 	& \ctrl{1}   & \targ{}	 & \ctrl{1}  & \qw \\
			& \targ{} 	  &\ctrl{-1}  & \targ{} 	& \targ{}    & \ctrl{-1} & \targ{}   & \qw
	\end{quantikz}}
	\]
	\[
	\overset{\text{Rule~\ref{rule-gateswap}}}{\Longleftrightarrow} 
	\scalebox{0.75}{\begin{quantikz}[row sep={0.7cm,between origins}, column sep=0.4cm]
			\setwiretype{n}	
			& \textcolor{red}{\push{M_4}}  &\textcolor{red}{\push{M_5}} & \textcolor{red}{\push{M_4}} 	& \push{M_6} & \push{M_5}& \push{M_2}& \\
			& \ctrl{1} 	  &\ctrl{1}   & \ctrl{1} 	& \ctrl{1} 	 & \ctrl{1}	 & \octrl{1} & \qw \\
			& \ctrl{1}    &\targ{}	  & \ctrl{1} 	& \octrl{1}  & \targ{}	 & \ctrl{1}  & \qw \\
			& \targ{} 	  &\ctrl{-1}  & \targ{} 	& \targ{}    & \ctrl{-1} & \targ{}   & \qw
	\end{quantikz}}
	\]
	\[
	\overset{\text{Lemma~\ref{lem-ABA-BAB}}}{\Longleftrightarrow}
	\scalebox{0.75}{\begin{quantikz}[row sep={0.7cm,between origins}, column sep=0.4cm]
			\setwiretype{n}	
			&\push{M_5} & \push{M_4} &\textcolor{blue}{\push{M_5}} & \textcolor{blue}{\push{M_6}} & \textcolor{blue}{\push{M_5}} & \push{M_2}& \\
			&\ctrl{1}   & \ctrl{1} 	 &\ctrl{1} 	& \ctrl{1} 	 & \ctrl{1}	 & \octrl{1} & \qw \\
			&\targ{}	& \ctrl{1} 	 &\targ{}	& \octrl{1}  & \targ{}	 & \ctrl{1}  & \qw \\
			&\ctrl{-1}  & \targ{} 	 &\ctrl{-1} & \targ{}    & \ctrl{-1} & \targ{}   & \qw
	\end{quantikz}}
	\]
	\[
	\overset{\text{Lemma~\ref{lem-ABA-BAB}}}{\Longleftrightarrow}
	\scalebox{0.75}{\begin{quantikz}[row sep={0.7cm,between origins}, column sep=0.4cm]
			\setwiretype{n}	
			&\push{M_5} & \textcolor{red}{\push{M_4}} & \textcolor{red}{\push{M_6}} & \push{M_5}  & \push{M_6} & \push{M_2}& \\
			&\ctrl{1}   & \ctrl{1} 	 & \ctrl{1}   & \ctrl{1}	& \ctrl{1}	 & \octrl{1} & \qw \\
			&\targ{}	& \ctrl{1} 	 & \octrl{1}  & \targ{}		& \octrl{1}  & \ctrl{1}  & \qw \\
			&\ctrl{-1}  & \targ{} 	 & \targ{}    & \ctrl{-1} 	& \targ{}	 & \targ{}   & \qw
	\end{quantikz}}
	\]
	\[
	\overset{\text{Rule~\ref{rule-gateswap}}}{\Longleftrightarrow}
	\scalebox{0.75}{\begin{quantikz}[row sep={0.7cm,between origins}, column sep=0.4cm]
			\setwiretype{n}	
			&\push{M_5} & \push{M_6}\slice{} & \push{M_4} & \push{M_5}  & \push{M_6}\slice{} & \push{M_2}& \\
			&\ctrl{1}   & \ctrl{1} 	 & \ctrl{1}   & \ctrl{1}	& \ctrl{1}	 & \octrl{1} & \qw \\
			&\targ{}	& \octrl{1}  & \ctrl{1}   & \targ{}		& \octrl{1}  & \ctrl{1}  & \qw \\
			&\ctrl{-1}  & \targ{} 	 & \targ{}    & \ctrl{-1} 	& \targ{}	 & \targ{}   & \qw
	\end{quantikz}}
	\]
\end{example}

\begin{theorem}[\textbf{Completeness}]
	If $\mathbf{A}\equiv \mathbf{B}$, then $\mathbf{A} \Leftrightarrow \mathbf{B}$.
\end{theorem}
\begin{proof}
	Let $\mathbf{A}, \mathbf{B}$ be two reversible circuits such that $\mathbf{A}\equiv \mathbf{B}$. By Theorem~\ref{thm-circuitcanonical}, there is a unique reversible circuit $\mathbf{C}$ in canonical form based on $\mathbb{H}$ such that $\mathbf{A} \Leftrightarrow \mathbf{C}$ and $\mathbf{B} \Leftrightarrow \mathbf{C}$. It follows immediately that $\mathbf{A} \Leftrightarrow \mathbf{B}$.
\end{proof}

\section{Conclusion and discussion}\label{sec-conclusion}
In this paper, we present the first complete set $\mathcal{RC}$ of transformation rules for  reversible circuits. To prove the completeness, we define the canonical forms of $n$-bit reversible circuits based on a Hamiltonian path of an $n$-hypercube graph, and show that every reversible function is computed by a unique reversible circuit in the canonical form. Moreover, we show that every reversible circuit can be transformed into its canonical form by applying the rules. Therefore, any two equivalent reversible circuits can be transformed into one another through the canonical form. 
This result can be utilized to demonstrate the completeness of a rule-based optimization system. That is, if the system encompasses the rules in $\mathcal{RC}$, it is theoretically capable of achieving circuit optimality. Nevertheless, this guarantee holds only in theory. As indicated by Algorithm~\ref{algo-contruction}, the canonical circuit of a reversible function can grow exponentially in size. Hence, converting a circuit into its canonical form may require an exponential number of transformation steps, rendering the optimization process impractically inefficient for large-scale circuits. To mitigate this limitation, heuristic strategies can be adopted—for example, new rules and templates can be derived from $\mathcal{RC}$ to reduce the overall transformation cost.

It should be noted that while the rules presented in this paper can be applied to optimize gate counts, a reduction at the logical level does not necessarily imply lower physical implementation costs. This work is primarily a theoretical contribution to reversible logic synthesis, specifically the completeness of transformation rules. The theory itself is not dedicated to gate minimization, nor does it address physical implementation. Furthermore, reversible logic gates are not exclusive to quantum computing and can be realized within other models of reversible computation.

In this work, we focus on the transformation rules for reversible circuits without ancillary bits.
There are two main reasons: first, the circuits without ancillary bits are already functionally complete; second, the current quantum hardware is constrained by the limited number of reliable qubits. Therefore, the use of qubits should be minimized when designing quantum algorithms. 
In~\cite{Wu2024asymptotically}, a new synthesis method is presented that can achieve the lower bound $O(2^n n/\log n)$ when no ancillary bit is used. Given access to a single ancillary bit (not necessarily constant), any MCT gate can be decomposed into a cascade of Toffoli gates~\cite{miller2011elementary,barenco1995,Aaronson2015classification}. This decomposition implies that all reversible circuits can be synthesized using only the X, CNOT, and Toffoli gates augmented by one ancillary bit, for which the complete transformation rules are provided by~\cite{Comfort2018category}. However, when introducing additional ancillary bits is prohibited, the complete axiomatization for the reversible circuits with ancillary bits is still unknown. 

Another question is the minimality of $\mathcal{RC}$, that is, whether the five rules are independent of each other. In particular, Rule~\ref{rule-completeness} implies Rule~\ref{rule-MPMCTtoMCT}, which is widely used for circuit transformation and has a more concise form. Can we replace Rule~\ref{rule-completeness} by Rule~\ref{rule-MPMCTtoMCT} so that the new theory still preserves completeness? This is also desirable for future research.

\bibliographystyle{IEEEtran}
\bibliography{./quantumref}

%

\begin{IEEEbiography}[{\includegraphics[width=1in,height=1.25in,clip,keepaspectratio]{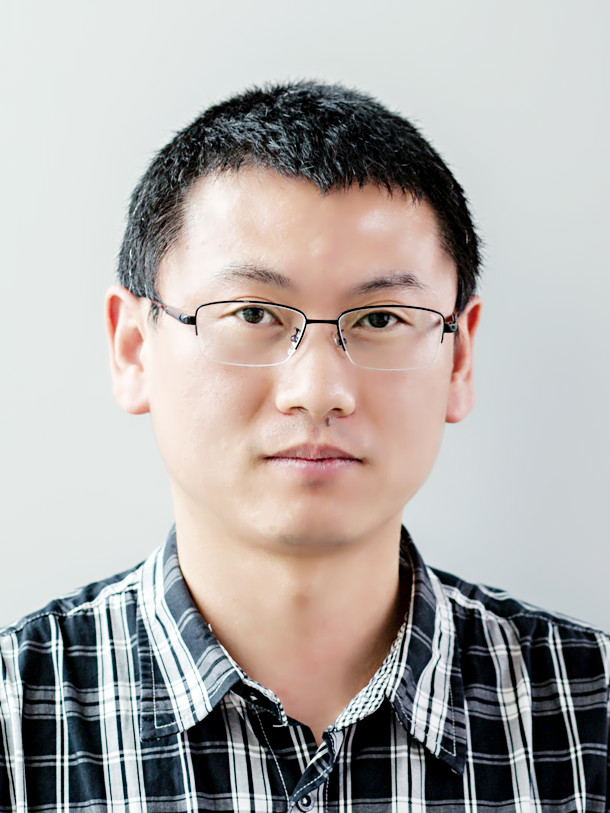}}]{Shiguang Feng}
	(Member, IEEE) received the B.S. degree in computer science and technology from Shandong Agricultural University, Tai'an, China, in 2006; the Ph.D. degree in logic from Sun Yat-sen University, Guangzhou, China, in 2012; and the Doctor of Natural Science degree in computer science from Leipzig University, Leipzig, Germany, in 2016. He is currently an associate researcher with the School of Computer Science and Engineering, Sun Yat-sen University, Guangzhou, China. His current research interests include reversible logic synthesis, quantum algorithms, and mathematical logic.
\end{IEEEbiography}

\begin{IEEEbiography}[{\includegraphics[width=1in,height=1.25in,clip,keepaspectratio]{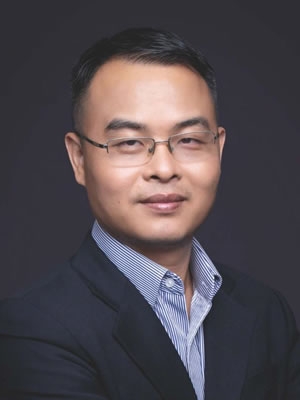}}]{Lvzhou Li}
	received his PhD degree in Computer Science from Sun Yat-sen University, China in 2009 and then worked in Sun Yat-sen University, China. Now he is a professor of the School of Computer Science and Engineering, Sun Yat-sen University, China. His research interests are quantum algorithm,  quantum circuit synthesis and optimization, and quantum machine learning.
\end{IEEEbiography}
\end{document}